\documentclass[11pt]{article}

\usepackage{fullpage}
\usepackage{amsmath,amsthm,amssymb,amsfonts}
\usepackage{array}
\usepackage[linesnumbered,vlined,procnumbered]{algorithm2e}
\newtheorem{theorem}{Theorem}[section]
\newtheorem{corollary}[theorem]{Corollary}
\newtheorem{observation}[theorem]{Observation}
\newtheorem{lemma}[theorem]{Lemma}
\newtheorem{definition}[theorem]{Definition}
\newtheorem*{remark}{Remark}

\DontPrintSemicolon
\SetKwComment{tcpy}{\#~}{} 
\SetKwFor{Perform}{perform}{times}{endp} 
\SetKwFor{For}{for}{do}{ends} 
\SetKwFor{While}{while}{do}{ends} 
\SetKwFor{Use}{use}{to}{endu} 
\SetKwInput{KwVars}{Variables} 
\SetKwInput{KwInput}{Input} 
\SetKwInput{KwOutput}{Output}
\SetKwInput{KwSet}{Variables set} 
\RestyleAlgo{boxruled}
\LinesNumbered

\newcommand{\accept}{\texttt{accept}}
\newcommand{\reject}{\texttt{reject}}
\newcommand{\restrict}[2]{{#1}\upharpoonright_{#2}}
\newcommand{\poly}{\mathrm{poly}}
\newcommand{\bias}{\mathrm{bias}}

\title{Improving and extending the testing of distributions for shape-restricted properties}
\author{Eldar Fischer\thanks{Faculty of Computer Science, Israel Institute of Technology (Technion),
  Haifa, Israel. \mbox{eldar@cs.technion.ac.il}} \and Oded Lachish \thanks{Birkbeck, University of
    London, London, UK. \mbox{oded@dcs.bbk.ac.uk}} \and Yadu Vasudev \thanks{Faculty of Computer
      Science, Israel Institute of Technology (Technion), Haifa, Israel.
      \mbox{yaduvasudev@gmail.com}}}

\begin{document}
\maketitle

\begin{abstract}
  Distribution testing deals with what information can be deduced about an unknown distribution over $\{1,\ldots,n\}$,
  where the algorithm is only allowed to obtain a relatively small number of independent samples from the distribution. In the extended conditional sampling model, the algorithm is also allowed to obtain samples from the restriction of the original distribution on subsets of $\{1,\ldots,n\}$.
  
  In 2015, Canonne, Diakonikolas, Gouleakis and Rubinfeld unified several previous results, and showed that for any property of distributions satisfying a ``decomposability'' criterion, there exists an algorithm (in the basic model) that can distinguish with high probability distributions satisfying the property from distributions that are far from it in the variation distance.
  
  We present here a more efficient yet simpler algorithm for the basic model, as well as very efficient algorithms for the conditional model, which until now was not investigated under the umbrella of decomposable properties. Additionally, we provide an algorithm for the conditional model that handles a much larger class of properties.
  
  Our core mechanism is a way of efficiently producing an interval-partition of $\{1,\ldots,n\}$ that satisfies a ``fine-grain'' quality. We show that with such a partition at hand we can directly move forward with testing individual intervals, instead of first searching for the ``correct'' partition of $\{1,\ldots,n\}$.
\end{abstract}

\section{Introduction}

\subsection{Historical background}

In most computational problems that arise from modeling real-world situations, we are required to
analyze large amounts of data to decide if it has a fixed property.  The amount of data involved is
usually too large for reading it in its entirety, both with respect to time and
storage. In such situations, it is natural to ask for algorithms that can sample points from the data
and obtain a significant estimate for the property of interest.  
The area of {\em
property testing} addresses this issue by studying algorithms that look at a small part of the data
and then decide if the object that generated the data has the property or is far (according to some
metric) from having the property. 

There has been a long line of research, especially in statistics, where the underlying object from
which we obtain the data is modeled as a probability distribution. Here the algorithm is only
allowed to ask for independent samples from the distribution, and has to base its decision on them.
If the support of the underlying probability distribution is large, it is not practical to
approximate the entire distribution. 
Thus, it is natural to study this problem in the context of property testing.

The specific sub-area of property testing that is dedicated to the study of distributions is called {\em distribution testing}.
There, the input is a
probability distribution (in this paper the domain is the set $\{1,2,\ldots, n\}$) and the objective
is to distinguish whether the distribution has a certain property, such as uniformity or
monotonicity, or is far in $\ell_1$ distance from it. See \cite{Canonne15} for a survey about the
realm of distribution testing.

%
%
Testing properties of distributions was studied by Batu et al in \cite{BatuFRSW00}, where they gave a
sublinear query algorithm for testing closeness of distributions supported over the set $\{1,2,
\ldots, n\}$. 
They extended the idea of collision counting,  which was implicitly used for uniformity testing in the work of
Goldreich and Ron (\cite{GoldreichR00}). Consequently, various properties of probability
distributions were studied, like testing identity with a known distribution
(\cite{BatuFFKRW01,ValiantV14, AcharyaDK15, DiakonikolasK16}), testing independence of a
distribution over a product space (\cite{BatuFFKRW01, AcharyaDK15}), and testing $k$-wise
independence (\cite{AlonAKMR07}).

In recent years, distribution testing has been extended beyond the classical model. A new model called the {\em conditional  sampling model} was introduced.
It first appeared independently in~\cite{CanonneRS15} and~\cite{ChakrabortyFGM13}. 
In the conditional sampling model, the algorithm
queries the input distribution $\mu$ with a set $S \subseteq \{1, 2, \ldots, n\}$, and gets an index
sampled according to $\mu$ conditioned on the set $S$. Notice that if $S = \{1, 2, \ldots, n\}$,
then this is exactly like in the standard model. 
The conditional sampling model allows adaptive
querying of $\mu$, since we can choose the set $S$ based on the indexes sampled until now.
Chakraborty et al (\cite{ChakrabortyFGM13}) and Canonne et al (\cite{CanonneRS15}) showed that testing
uniformity can be done with a number of queries not depending on $n$ (the latter presenting an
optimal test), and investigated the testing of other properties of distributions.  
In~\cite{ChakrabortyFGM13}, it is also shown that uniformity can be tested with $\poly(\log n)$
conditional samples by a
{\em non-adaptive} algorithm.  In this work, we study the distribution testing problem in the
standard (unconditional) sampling model, as well as in the conditional model.

A line of work which is central to our paper, is the testing of distributions for {\em structure}.
The objective is to test whether a given distribution has some structural properties like being
monotone (\cite{BatuKR04}), being a $k$-histogram (\cite{IndykLR12, DiakonikolasK16}), or being
log-concave (\cite{AcharyaDK15}). 
Canonne et al (\cite{CanonneDGR15}) unified these results to show
that if a property of distributions has certain structural characteristics, then membership in the property can
be tested efficiently using samples from the distribution. 
More precisely, they introduced the notion of {\em $L$-decomposable} distributions as a
way to unify various algorithms for testing distributions for structure. Informally, an
$L$-decomposable distribution $\mu$ supported over $\{1, 2, \ldots, n\}$ is one that has an
interval partition $\mathcal{I}$ of $\{1, 2, \ldots, n\}$ of size bounded by $L$, such that for
every interval $I$, either the weight of $\mu$ on it is small or the distribution over the interval
is close to uniform. A property $\mathcal{C}$ of distributions is $L$-decomposable if every
distribution $\mu \in \mathcal{C}$ is $L$-decomposable ($L$ is allowed to depend on $n$).  This
generalizes various properties of distributions like being monotone, unimodal, log-concave etc. 
In this setting, their result for a set of distributions $\mathcal{C}$ supported over $\{1,2,\ldots, n\}$ translates to the following: if every distribution $\mu$ from
$\mathcal{C}$ is $L$-decomposable, then there is an efficient algorithm for
testing whether a given distribution belongs to the property $\mathcal{C}$.

To achieve their results,
Canonne et al (\cite{CanonneDGR15}) show that if a distribution $\mu$ supported over $[n]$
is $L$-decomposable, then it is $O(L \log n)$-decomposable where the intervals are of the form
$[j2^i +1, (j + 1)2^i]$. This presents a natural approach of computing the
interval partition in a recursive manner, by bisecting an interval if it has a large probability
weight and is not close to uniform. Once they get an interval partition, they learn the
``flattening'' of the distribution over this partition, and check if this distribution is close to
the property $\mathcal{C}$. The term ``flattening'' refers to the distribution resulting from making $\mu$ conditioned on any interval of the partition to be uniform. When applied to a partition corresponding to a decomposition of the distribution, the learned flattening is also close to the original distribution.
Because of this, in the case where there is a promise that $\mu$ is $L$-decomposable, the above can be viewed as a learning algorithm, where they obtain an
explicit distribution that is close to $\mu$. Without the promise it can be viewed as an agnostic learning algorithm. For further elaboration of this connection see~\cite{diakonikolas2016learning}.

\subsection{Results and techniques}

In this paper, we extend the body of knowledge about testing $L$-decomposable properties. We improve
upon the previously known bound on the sample complexity, and give much better bounds when
conditional samples are allowed. Additionally, for the conditional model, we provide a test for a
broader family of properties, that we call {\em atlas-characterizable} properties.

Our approach differs from that of \cite{CanonneDGR15} in the manner in which we compute the interval
partition. We show that a partition where most intervals that are not singletons have small
probability weight is sufficient to learn the distribution $\mu$, even though it is not the original
$L$-decomposition of $\mu$.
We show
that if a distribution $\mu$ is $L$-decomposable, then the ``flattening'' of $\mu$ with respect to a
this partition is close to $\mu$. It turns out that such a partition can be obtained in ``one shot''
without resorting to a recursive search procedure.

We obtain a partition as above using a method of {\em partition pulling}
that we develop here.  Informally, a pulled partition is obtained by sampling indexes from $\mu$,
and taking the partition induced by the samples in the following way: each sampled index is a
singleton interval, and the rest of the partition is composed of the maximal intervals between sampled
indexes. Apart from the obvious simplicity of this procedure, it also has the advantage of providing a
partition with a significantly smaller number of intervals, linear in $L$ for a fixed $\epsilon$, and with
no dependency on $n$ unless $L$ itself depends on it. This makes our algorithm more efficient in query
complexity than the one of \cite{CanonneDGR15} in the unconditional sampling model, and leads to a
dramatically small sampling complexity in the (adaptive) conditional model.

Another feature of the partition pulling method is that it provides a partition with small weight intervals also when the distribution is not $L$-decomposable. This allows us to use the partition in a different manner later on, in the algorithm for testing atlas characterizable properties using conditional samples.

The main common ground between our approach for $L$-decomposable properties and that of \cite{CanonneDGR15} is the method of testing by implicit learning, as defined formally in \cite{DiakonikolasLMORSW07} (see \cite{Servedio10}). In particular, the results also provide a means to learn a distribution close to $\mu$ if $\mu$ satisfies the tested property. We also provide a test under the conditional query model for the extended class of atlas characterizable properties that we define below, which generalizes both decomposable properties and symmetric properties. A learning algorithm for this class is not provided; only an ``atlas'' of the input distribution rather than the distribution itself is learned.

Our result for unconditional testing (Theorem \ref{thm:learn}) gives a $\sqrt{nL} / \poly(\epsilon)$
query algorithm in the standard (unconditional) sampling model for testing an
$L$-decomposable property of distributions. Our method of finding a good partition for $\mu$ using
pulled partitions, that we explained above, avoids the $\log n$ factor present in
Theorem 3.3 of~\cite{CanonneDGR15}. We also avoid the additional $O(L^2)$ additive term present there.
The same method enables us to extend our results to the conditional query model, which we present for
both adaptive and non-adaptive algorithms. Table \ref{tab:results} summarizes our results and compares
them with known lower bounds\footnote{The lower bounds for unconditional and non-adaptive conditional testing
of $L$-decomposable properties with $L=1$ are exactly the lower bounds for uniformity testing; the lower bound
for adaptive conditional testing follows easily from the proved existence of properties that have no sub-linear
complexity adaptive conditional tests; finally, the lower bound for properties $k$-characterized by atlases with
$k=1$ is just a bound for a symmetric property constructed there. About the last one, we conjecture that there
exist properties with much higher lower bounds.}.

\begin{table}[h]
  \caption{Summary of our results}
  \label{tab:results}
  \begin{center}
    \begin{tabular}{|c|m{4cm}|m{6cm}|}
  \hline
   & Result & Known lower bound \\
  \hline
  \textbf{$L$-decomposable} & (testing and learning) & \\
  \hline
  Unconditional & $\sqrt{nL}/\poly(\epsilon)$ & $\Omega(\sqrt{n}/\epsilon^2)$ for $L =
  1$ \cite{Paninski08}\\
  \hline
  Adaptive conditional & $L / \poly(\epsilon)$ & $\Omega(L)$ for some fixed $\epsilon$
  \cite{ChakrabortyFGM13}\\
  \hline
  Non-adaptive conditional & $L \cdot \poly(\log n, 1/\epsilon)$ & $\Omega(\log n)$ for $L=1$ and some fixed
  $\epsilon$
  \cite{AcharyaCK15} \\
  \hline
  \textbf{$k$-characterized by atlases} & (testing) & \\
  \hline
  Adaptive conditional & $k \cdot \poly(\log n, 1/\epsilon)$ & $\Omega(\sqrt{\log\log n})$ for $k=1$, and some
  fixed $\epsilon$ \cite{ChakrabortyFGM13}\\
  \hline
\end{tabular}
\end{center}
\end{table}


 \section{Preliminaries}\label{sec:prelim}
 
 We denote the set $\{1,\ldots,n\}$ by $[n]$.
 
 We study the problem of testing properties of probability distributions supported over
 $[n]$, when given samples from the distribution. For two distributions $\mu$ and
 $\chi$, we say that $\mu$ is \emph{$\epsilon$-far} from $\chi$ if they are far in the $\ell_1$ norm, that is,
 $d(\mu,\chi)=\sum_{i\in [n]} |\mu(i) - \chi(i)| > \epsilon$. For a property of distributions $\mathcal{C}$,
 we say that $\mu$ is  \emph{$\epsilon$-far} from $\mathcal{C}$
 if for all $\chi \in \mathcal{C}$, $d(\mu,\chi) > \epsilon$.
 
 Outside the $\ell_1$ norm between distributions, we also use the $\ell_\infty$ norm, $\lVert \mu - \chi \rVert_\infty = \max_{i\in [n]}|\mu(i)-\chi(i)|$, and the following measure for uniformity.

 \begin{definition}
  For a distribution $\mu$ over a domain $I$, we define the {\em bias} of $\mu$ to be $\bias(\mu) = \max_{i \in I} \mu(i) / \min_{i \in I} \mu(i) - 1$.
 \end{definition}

 The following observation is easy and will be used implicitly throughout.
 \begin{observation}\label{obs:norms}
  For any two distributions $\mu$ and $\chi$ over a domain $I$ of size $m$, $d(\mu,\chi)\leq m\lVert \mu - \chi \rVert_\infty$. Also, $\lVert \mu - \mathcal{U}_I \rVert_\infty \le \frac1m\bias(\mu)$, where $\mathcal{U}_I$ denotes the uniform distribution over $I$.
 \end{observation}
 \begin{proof}
  Follows from the definitions.
 \end{proof}

 We study the problem, both in the standard model, where the algorithm is given indexes sampled
 from the distribution, as well as in the model of conditional samples. The conditional model was
 first studied in the independent works of Chakraborty et al (\cite{ChakrabortyFGM13}) and Canonne et
 al (\cite{CanonneRS15}). We first give the definition of a conditional oracle for a distribution
 $\mu$.
 
 \begin{definition} [Conditional oracle]
 A conditional oracle to a distribution $\mu$ supported over $[n]$ is a black-box that takes as input
 a set $A\subseteq [n]$, samples a point $i\in A$ with probability $\mu(i)/\sum_{j\in A}\mu(j)$,
 and returns $i$. If $\mu(j) = 0$ for all $j \in A,$ then it chooses $i\in A$ uniformly at random.
 \end{definition}
 
 \begin{remark}
    The behaviour of the conditional oracle on sets $A$ with $\mu(A)=0$ is as per the model of
    Chakraborty et al \cite{ChakrabortyFGM13}. However, upper bounds in this model also hold in the
    model of Canonne et al \cite{CanonneRS15}, and most lower bounds can be easily converted to it.
 \end{remark}
 
 Now we define conditional distribution testing algorithms. We will define and analyze both {\em
 adaptive} and {\em non-adaptive} conditional testing algorithms.
 
 \begin{definition}
   An {\em adaptive} conditional distribution testing algorithm for a property of distributions
   $\mathcal{C}$, with parameters $\epsilon, \delta > 0$, and $n \in \mathbb{N}$, with query
   complexity $q(\epsilon,\delta,n)$, is a randomized algorithm with access to a conditional oracle
   of a distribution $\mu$ with the following properties:
   \begin{itemize}
   \item For each $i\in [q]$, at the $i^{th}$ phase, the algorithm generates a set $A_i\subseteq
     [n]$, based on $j_1,j_2,\cdots,j_{i-1}$ and its internal coin tosses, and calls the conditional
     oracle with $A_i$ to receive an element $j_i$, drawn independently of $j_1, j_2,\cdots,
     j_{i-1}$.
 \item Based on the received elements $j_1, j_2, \cdots, j_q$ and its internal coin tosses, the
   algorithm accepts or rejects the distribution $\mu$.
   \end{itemize}
 If $\mu\in \mathcal{C}$, then the algorithm accepts with probability at least $1-\delta$, and if
 $\mu$ is $\epsilon$-far from $\mathcal{C}$, then the algorithm rejects with probability at least
 $1-\delta$.
 \end{definition}
 
 \begin{definition}
   A {\em non-adaptive} conditional distribution testing algorithm for a property of distributions
   $\mathcal{C}$, with parameters $\epsilon, \delta > 0$, and $n\in \mathbb{N}$, with query
   complexity $q(\epsilon, \delta, n)$, is a randomized algorithm with access to a conditional oracle of
   a distribution $\mu$ with the following properties:
   \begin{itemize}
     \item The algorithm chooses sets $A_1,\ldots,A_q$ (not necessarily distinct) based on its
       internal coin tosses, and then queries the conditional oracle to respectively obtain
       $j_1,\ldots,j_q$.
     \item Based on the received elements $j_1,\ldots,j_q$ and its internal coin tosses, the
       algorithm accepts or rejects the distribution $\mu$.
   \end{itemize}
   If $\mu\in \mathcal{C}$, then the algorithm accepts with probability at least $1 - \delta$, and if
   $\mu$ is $\epsilon$-far from $\mathcal{C}$, then the algorithm rejects with probability at least
   $1 - \delta$.
 \end{definition}

 \subsection{Large deviation bounds}
 
 The following large deviation bounds will be used in the analysis of our algorithms through the rest
 of the paper.
 
 \begin{lemma}[Chernoff bounds]
   Let $X_1, X_2, \ldots, X_m$ be independent random variables taking values in $\{0,1\}$. Let $X =
   \sum_{i \in [m]} X_i$. Then for any $\delta \in
 (0,1]$, we have the following bounds.
 \begin{enumerate}
   \item $\Pr[X > (1 + \delta)E[X]] \le \exp\left( \tfrac{-\delta^2 E[X]}{3} \right)$.
   \item $\Pr[X < (1 - \delta)E[X]] \le \exp\left( \tfrac{-\delta^2 E[X]}{2} \right)$.
 \end{enumerate}
 When $\delta \ge 1$, we have $\Pr[X \ge (1 + \delta)E[X]] < \exp(\tfrac{-\delta E[X]}{3})$.
   \label{lem:chernoff}
 \end{lemma}
 
 \begin{lemma}
   [Hoeffding bounds \cite{Hoeffding63}]
   Let $X_1, X_2, \ldots, X_m$ be independent random variables such that $0 \le X_i \le 1$ for $i
   \in [m]$, and let $\overline{X} = \tfrac{1}{m}\sum_{i \in [m]} X_i$. Then
     $\Pr\left[ |\overline{X} - E[\overline{X}]| > \epsilon \right] \le 2 \exp\left( -2m\epsilon^2
     \right)$
   \label{lem:hoeffding}
 \end{lemma}
 
 \subsection{Basic distribution procedures}
 
 The following is a folklore result about learning any distribution supported over $[n]$, that we
prove here for completeness. 

\begin{lemma}[Folklore]
  Let $\mu$ be a distribution supported over $[n]$. Using $\tfrac{n + \log(2/\delta)}{2\epsilon^2}$
  unconditional samples from $\mu$, we can obtain an explicit distribution $\mu'$ supported on $[n]$
  such that, with probability at least $1 - \delta$, $d(\mu, \mu') \le \epsilon$.
  \label{lem:learn-dist}
\end{lemma}
\begin{proof}
  Take $m = \tfrac{n+\log(2/\delta)}{2\epsilon^2}$ samples from $\mu$, and for each $i \in [n]$,
  let $m_i$ be the number of times $i$ was sampled. Define $\mu'(i) = m_i/m$.  Now, we show that
  $\max_{S \subseteq [n]}|\mu(S) - \mu'(S)| \le \epsilon/2$.  The lemma follows from this since the
  $\ell_1$ distance is equal to twice this amount.

  For any set $S \subseteq [n]$, let $X_1, X_2, \ldots, X_m$ be random variables such that $X_j = 1$
  if the $j^{th}$ sample was in $S$, and otherwise $X_j = 0$. Let $\overline{X} =
  \tfrac{1}{m}\sum_{j \in [m]} X_j$. Then, $\overline{X} = \mu'(S)$ and $E[\overline{X}] = \mu(S)$.
  By Lemma \ref{lem:hoeffding}, $\Pr[|\overline{X} - E[\overline{X}]| > \epsilon/2] \le
  2e^{-2m\epsilon^2}$. Substituting for $m$, we get that $\Pr[|\mu'(S) - \mu(S)| > \epsilon/2] \le
  2e^{-n - \log(2/\delta)}$. Taking a union bound over all sets, with probability at least $1 -
  \delta$, $|\mu'(S) - \mu(S)| \le \epsilon/2$ for every $S \subseteq [n]$. Therefore, $d(\mu, \mu')
  \le \epsilon$.
\end{proof}

We also have the following simple lemma about learning a distribution in $\ell_\infty$ distance.
\begin{lemma}
  Let $\mu$ be a distribution supported over $[n]$. Using $\tfrac{\log(2n/\delta)}{2\epsilon^2}$
  unconditional samples from $\mu$, we can obtain an explicit distribution $\mu'$ supported on $[n]$
  such that, with probability at least $1 - \delta$, $\lVert \mu - \mu' \rVert_\infty \le \epsilon$.
  \label{lem:learn-infty}
\end{lemma}
\begin{proof}
  Take $m = \tfrac{\log(2n/\delta)}{2\epsilon^2}$ samples from $\mu$, and for each $i \in [n]$, let
  $m_i$ be the number of times $i$ was sampled. For each $i \in [n]$, define $\mu'(i) = m_i/m$.

  For an index $i \in [n]$, let $X_1, X_2, \ldots, X_m$ be random variables such that $X_j = 1$ if
  the $j^{th}$ sample is $i$, and otherwise $X_j = 0$. Let $\overline{X} = \tfrac{1}{m} \sum_{j \in
  [m]} X_j$. Then $\overline{X} = \mu'(i)$ and $E[\overline{X}] = \mu(i)$. By Lemma
  \ref{lem:hoeffding}, $\Pr[|\overline{X} - E[\overline{X}]| > \epsilon] \le 2e^{-2m\epsilon^2}$.
  Substituting for $m$, we get that $\Pr[|\mu(i) - \mu'(i)| > \epsilon] \le \delta/n$. By a union
  bound over $[n]$, with probability at least $1 - \delta$, $|\mu(i) - \mu'(i)| \le \epsilon$ for
  every $i \in [n]$.
\end{proof}


\section{Fine partitions and how to pull them} \label{sec:pullfine}
 
 We define the notion of $\eta$-fine partitions of a distribution $\mu$ supported over $[n]$, which
 are central to all our algorithms.

 \begin{definition}
   [$\eta$-fine interval partition]
   Given a distribution $\mu$ over $[n]$, an {\em $\eta$-fine} interval partition of $\mu$ is an
   interval-partition $\mathcal{I} = (I_1, I_2, \ldots, I_r)$ of $[n]$ such that for all $j \in [r]$,
   $\mu(I_j) \le \eta$, excepting the case $|I_j| = 1$.  
   The \emph{length} $|\mathcal{I}|$ of an interval partition $\mathcal{I}$ is the number of intervals in it.
 \end{definition}
 
 The following Algorithm \ref{alg:pulled-partition} is the pulling mechanism. The idea is to take independent unconditional samples from $\mu$, make them into singleton intervals in our interval-partition $\mathcal{I}$, and then take the intervals between these samples as the remaining intervals in $\mathcal{I}$.
 
 \begin{algorithm}
   [htpb]
   \caption{Pulling an $\eta$-fine partition}
   \label{alg:pulled-partition}
 
   \KwInput{Distribution $\mu$ supported over $[n]$, parameters $\eta>0$ (fineness) and $\delta > 0$ (error probability)}

 Take $m = \tfrac{3}{\eta}\log\left( \tfrac{3}{\eta \delta} \right)$ unconditional samples
 from $\mu$
 \label{step:pulled-sample}
 
 Arrange the indices sampled in increasing order $i_1 < i_2 < \cdots < i_r$ without repetition and set $i_0=0$
 \label{step:pulled-arrange}
 
 \For {each $j\in [r]$}
 {
 \lIf {$i_j  > i_{j-1} +1$}
 {add the interval $\{i_{j-1} + 1, \ldots, i_j - 1\}$ to $\mathcal{I}$}
 \label{step:pulled-interval}
 
 Add the singleton interval $\{i_j\}$ to $\mathcal{I}$
 \label{step:pulled-singleton}
 }
 
 \lIf {$i_r < n$} 
 {add the interval $\{i_r + 1, \ldots, n\}$ to $\mathcal{I}$}
 \label{step:pulled-last-interval}

 \Return{$\mathcal{I}$}
 
 \end{algorithm}
 
 \begin{lemma}\label{lem:pulled-partition}
 Let $\mu$ be a distribution that is supported over $[n]$,  and $\eta,\delta > 0$, and
 suppose that these are fed to Algorithm~\ref{alg:pulled-partition}.
 Then, with probability at least $1 - \delta$, 
   the set of intervals $\mathcal{I}$ returned by  Algorithm~\ref{alg:pulled-partition} is an $\eta$-fine interval partition of $\mu$ of length $O\left( \tfrac{1}{\eta}\log\left(
   \tfrac{1}{\eta \delta} \right) \right)$.
 \end{lemma}
 \begin{proof}
   Let $\mathcal{I}$ the set of intervals returned by Algorithm~\ref{alg:pulled-partition}. 
   The guarantee on the length of $\mathcal{I}$ follows from the number of samples taken in Step~\ref{step:pulled-sample},
   noting that $|\mathcal{I}|\leq 2r-1=O(m)$.
   
   Let $\mathcal{J}$ be a maximal set of pairwise disjoint minimal intervals $I$ in $[n]$, such that  $\mu(I)\geq \eta/3$ for every interval $I\in \mathcal{J}$. Note that every $i$ for which $\mu(i)\geq\eta/3$ necessarily appears as a singleton interval $\{i\}\in\mathcal{J}$. Also clearly $|\mathcal{J}|\leq 3/\eta$.
   
   We shall first show that if an interval $I'$ is such that $\mu(I') \geq \eta$, then it fully contains some interval  $I\in \mathcal{J}$. Then, we shall show 
  that, with probability at least $1-\delta$,
   the samples taken in Step~\ref{step:pulled-sample} include an index from   every interval $I\in\mathcal{J}$.
   By Steps~\ref{step:pulled-arrange} to~\ref{step:pulled-last-interval} of the algorithm and the above, this implies the statement of the lemma.
     
   Let $I'$ be an interval such that $\mu(I')\geq \eta$, and assume on the contrary that it contains no interval from $\mathcal J$. Clearly it may intersect without containing at most two intervals $I_l,I_r\in\mathcal{J}$. Also, $\mu(I'\cap I_l)<\eta/3$ because otherwise we could have replaced $I_l$ with $I'\cap I_l$ in $\mathcal{J}$, and the same holds for $\mu(I'\cap I_r)$. But this means that $\mu(I\setminus(I_l\cup I_r))>\eta/3$, and so we could have added $I\setminus(I_l\cup I_r)$ to $\mathcal{J}$, again a contradiction.
   
   Let $I\in \mathcal{J}$.
   The probability that an index from $I$ is not sampled is at most $(1
   - \eta/3)^{3\log(3/\eta \delta)/\eta} \leq \delta \eta/3$. 
   By a union bound over all $I\in\mathcal{J}$, with probability at least $1 - \delta$ the samples taken in Step~\ref{step:pulled-sample} include an index from every interval in $\mathcal{J}$.   
 \end{proof}
 
 The following is a definition of a variation of a fine partition, where we allow some intervals of small total weight to violate the original requirements.
 
 \begin{definition}
   [$(\eta, \gamma)$-fine partitions]
   Given a distribution $\mu$ over $[n]$,
   an {\em $(\eta, \gamma)$-fine} interval partition is an interval
   partition $\mathcal{I} = (I_1, I_2, \ldots, I_r)$ such that $\sum_{I\in\mathcal{H}_{\mathcal{I}}}\mu(I) \le \gamma$, where $\mathcal{H}_{\mathcal{I}}$ is the set of {\em violating intervals} $\left\{ I \in \mathcal{I} : \mu(I) > \eta, |I| > 1 \right\}$.
   \label{defn:gamma-eta}
 \end{definition}
 
 In our applications, $\gamma$ will be larger than $\eta$ by a factor of $L$, which would allow us through the following Algorithm \ref{alg:pulled-rel-partition} to avoid having additional $\log L$ factors in our expressions for the unconditional and the adaptive tests.
 
 \begin{algorithm}
   [htpb]
   \caption{Pulling an $(\eta,\gamma)$-fine partition}
   \label{alg:pulled-rel-partition}
 
   \KwInput{Distribution $\mu$ supported over $[n]$, parameters $\eta, \gamma>0$ (fineness) and $\delta > 0$ (error probability)}

 Take $m = \tfrac{3}{\eta}\log\left( \tfrac{5}{\gamma\delta} \right)$ unconditional samples
 from $\mu$
 \label{step:pulled-rel-sample}
 
 Perform Step \ref{step:pulled-arrange} through Step \ref{step:pulled-last-interval} of Algorithm \ref{alg:pulled-partition}.

 \Return{$\mathcal{I}$}
 
 \end{algorithm}
 
 \begin{lemma}\label{lem:pulled-rel-partition}
 Let $\mu$ be a distribution that is supported over $[n]$,  and $\gamma, \eta,\delta > 0$, and
 suppose that these are fed to Algorithm~\ref{alg:pulled-rel-partition}.
 Then, with probability at least $1 - \delta$, 
   the set of intervals $\mathcal{I}$ returned by  Algorithm~\ref{alg:pulled-rel-partition} is an $(\eta,\gamma)$-fine interval partition of $\mu$ of length $O\left( \tfrac{1}{\eta}\log\left(\tfrac{1}{\gamma \delta} \right) \right)$.
 \end{lemma}
 \begin{proof}
   Let $\mathcal{I}$ the set of intervals returned by Algorithm~\ref{alg:pulled-rel-partition}. 
   The guarantee on the length of $\mathcal{I}$ follows from the number of samples taken in Step~\ref{step:pulled-rel-sample}.
   
   As in the proof of Lemma \ref{lem:pulled-partition}, Let $\mathcal{J}$ be a maximal set of pairwise disjoint minimal intervals $I$ in $[n]$, such that  $\mu(I)\geq \eta/3$ for every interval $I\in \mathcal{J}$. Here, also define the set $\mathcal{J}'$ to be the set of maximal intervals in $[n]\setminus\bigcup_{I\in\mathcal{J}}I$. Note now that $\mathcal{J}\cup\mathcal{J}'$ is an interval partition of $[n]$. Note also that between every two consecutive intervals of $\mathcal{J}'$ lies an interval of $\mathcal{J}$. Finally, since $\mathcal{J}$ is maximal, all intervals in $\mathcal{J}'$ are of weights less than $\eta/3$.
   
   Recalling the definition $\mathcal{H} = \left\{ I \in \mathcal{I} : \mu(I) >  \eta, |I| > 1 \right\}$, as in the proof of Lemma \ref{alg:pulled-partition} every $I'\in\mathcal{H}$ must fully contain an interval from $\mathcal{J}$ from which no point was sampled. Moreover, $I'$ may not fully contain intervals from $\mathcal{J}$ from which any points were sampled.
   
   Note furthermore that the weight of such an $I'$ is not more than $5$ times the total weight of the intervals in $\mathcal{J}$ that it fully contains. To see this, recall that the at most two intervals from $\mathcal{J}$ that intersect $I'$ without containing it have intersections of weight not more than $\eta/3$. Also, there may be the intervals of $\mathcal{J'}$ intersecting $I'$, each of weight at most $\eta/3$. However, because there is an interval $\mathcal{J}$ between any two consecutive intervals of $\mathcal{J}'$, the number of intervals from $\mathcal{J}'$ intersecting $I'$ is at most $1$ more than the number of intervals of $\mathcal{J}$ fully contained in $I'$. Thus the number of intersecting intervals from $\mathcal{J}\cup\mathcal{J}'$ is not more than $5$ times the number of fully contained intervals from $\mathcal{J}$, and together with their weight bounds we get the bound on $\mu(I')$.
   
   Let $I\in \mathcal{J}$. The probability that an index from $I$ is not sampled is at most $(1 - \eta/3)^{3\log(5/\gamma \delta)/\eta} \leq \delta \gamma/5$. 
   By applying the Markov bound over all $I\in\mathcal{J}$ (along with their weights), with probability at least $1 - \delta$ the samples taken in Step~\ref{step:pulled-sample} include an index from every interval in $\mathcal{J}$ but at most a subset of them of total weight at most $\gamma/5$. By the above this means that $\sum_{I\in\mathcal{H}}\mu(I) \le \gamma$. 
 \end{proof}
 

 \section{Handling decomposable distributions}\label{sec:finedecomp}
 
 The notion of $L$-decomposable distributions was defined and studied in \cite{CanonneDGR15}. They
 showed that a large class of properties, such as monotonicity and log-concavity, are
 $L$-decomposable. We now formally define $L$-decomposable distributions and properties, as given in
 \cite{CanonneDGR15}.
 
 \begin{definition}
   [$(\gamma,L)$-decomposable distributions \cite{CanonneDGR15}]
   For an integer $L$, a distribution $\mu$ supported over $[n]$ is $(\gamma,L)$-decomposable, if
   there exists an interval partition $\mathcal{I} = (I_1, I_2, \ldots, I_\ell)$ of $[n]$, where
   $\ell \le L$, such that for all $j \in [\ell]$, at least one of the following holds.
   \begin{enumerate}
     \item $\mu(I_j) \le \tfrac{\gamma}{L}$.
     \item $\max_{i \in I_j} \mu(i) \le (1 + \gamma) \min_{i \in I_j} \mu(i)$.
   \end{enumerate}
 \end{definition}
 
The second condition in the definition of a $(\gamma,L)$-decomposable distribution is identical to saying that $\bias(\restrict{\mu}{I_j})\leq \gamma$. An $L$-decomposable property is now defined in terms of all its members being decomposable distributions.

\begin{definition}
   [$L$-decomposable properties, \cite{CanonneDGR15}]
 For a function $L: (0,1] \times \mathbb N \to \mathbb N$, we say that a property of distributions
 $\mathcal{C}$ is $L$-decomposable, if for every $\gamma > 0$, and $\mu \in \mathcal{C}$ supported
 over $[n]$,  $\mu$ is $(\gamma,L(\gamma, n))$-decomposable.
 \end{definition}

Recall that part of the algorithm for learning such distributions is finding (through pulling) what we referred to as a fine partition.
Such a partition may still have intervals where the conditional distribution over them is far from uniform.
However, we shall show that for $L$-decomposable distributions, the total weight of such ``bad'' intervals is not very high.

 The next lemma shows that every fine partition of an $(\gamma, L)$-decomposable distribution
 has only a small weight concentrated on ``non-uniform'' intervals and thus it will be sufficient to deal with the ``uniform'' intervals.
 \begin{lemma}
   Let $\mu$ be a distribution supported over $[n]$ which is $(\gamma, L)$-decomposable. For every
   $\gamma/L$-fine interval partition $\mathcal{I}' = (I'_1, I'_2, \ldots, I'_r)$ of
   $\mu$, the following holds.
   \begin{align*}
     \sum_{j \in [r]: \bias(\restrict{\mu}{I'_j}) > \gamma} \mu(I'_j) \le 2\gamma.
   \end{align*}
   \label{lem:small-wt-non-uniform}
 \end{lemma}
 \begin{proof}
   Let $\mathcal{I} = (I_1, I_2, \ldots, I_\ell)$ be the $L$-decomposition of $\mu$, where
   $\ell \le L$. Let $\mathcal{I}' = (I'_1, I'_2, \ldots, I'_r)$ be an interval partition
   of $[n]$ such that for all $j \in [r]$, $\mu(I'_j) \le \gamma/L$ or $|I'_j| = 1$. 
 
   Any interval $I'_j$ for which
   $\bias(\restrict{\mu}{I'_j}) > \gamma$, is either completely inside an interval
   $I_k$ such that $\mu(I_k) \le \gamma/L$, or intersects more than one interval (and in particular
   $|I'_j| > 1$). There are at most $L-1$ intervals in $\mathcal{I}'$ that intersect more than one
   interval in $\mathcal{I}$.  The sum of the weights of all such intervals is at most $\gamma$.  
   
   For any interval $I_k$ of $\mathcal{I}$ such that $\mu(I_k) \le \gamma/L$, the sum of the weights
   of intervals from $\mathcal{I}'$ that lie completely inside $I_k$ is at most $\gamma/L$. Thus, the
   total weight of all such intervals is bounded by $\gamma$.
   Therefore, the sum of the weights of intervals $I'_j$ such that $\bias(\restrict{\mu}{I'_j}) > \gamma$ is at most $2\gamma$.
 \end{proof}

In order to get better bounds, we will use the counterpart of this lemma for the more general (two-parameter) notion of a fine partition.

 \begin{lemma}
   Let $\mu$ be a distribution supported over $[n]$ which is $(\gamma, L)$-decomposable. For every
   $(\gamma/L,\gamma)$-fine interval partition $\mathcal{I}' = (I'_1, I'_2, \ldots, I'_r)$ of
   $\mu$, the following holds.
   \begin{align*}
     \sum_{j \in [r]: \bias(\restrict{\mu}{I'_j}) > \gamma} \mu(I'_j) \le 3\gamma.
   \end{align*}
   \label{lem:small-wt-non-uniform-rel}
 \end{lemma}
 \begin{proof}
   Let $\mathcal{I} = (I_1, I_2, \ldots, I_\ell)$ be the $L$-decomposition of $\mu$, where
   $\ell \le L$. Let $\mathcal{I}' = (I'_1, I'_2, \ldots, I'_r)$ be an interval partition
   of $[n]$ such that for a set $\mathcal{H}_{\mathcal{I}}$ of total weight at most $\gamma$, for all $I'_j \in \mathcal{I}\setminus\mathcal{H}_{\mathcal{I}}$, $\mu(I'_j) \le \gamma/L$ or $|I'_j| = 1$.
 
   Exactly as in the proof of Lemma \ref{lem:small-wt-non-uniform}, the total weight of intervals $I'_j \in \mathcal{I}\setminus\mathcal{H}_{\mathcal{I}}$ for which $\bias(\restrict{\mu}{I'_j}) > \gamma$ is at most $2\gamma$. In the worst case, all intervals in $\mathcal{H}_{\mathcal{I}}$ are also  such that $\bias(\restrict{\mu}{I'_j}) > \gamma$, adding at most $\gamma$ to the total weight of such intervals.
 \end{proof}

As previously mentioned, we are not learning the actual distribution but a ``flattening'' thereof. We next formally define the flattening of a distribution $\mu$ with respect to an interval partition
$\mathcal{I}$. 
Afterwards we shall describe its advantages and how it can be learned.
\begin{definition}
  Given a distribution $\mu$ supported over $[n]$ and a partition $\mathcal{I} = (I_1,
  I_2, \ldots, I_\ell)$, of $[n]$ to intervals, the \emph{flattening} of $\mu$ with respect to $\mathcal{I}$ is a
  distribution $\mu_\mathcal{I}$, supported over $[n]$, such that for $i \in I_j$,
  $\mu_\mathcal{I}(i) = \mu(I_j)/|I_j|$.
\end{definition}

The following lemma shows that the flattening of any distribution $\mu$, with respect to any interval
partition that has only small weight on intervals far from uniform, is close to $\mu$.
\begin{lemma}
  Let $\mu$ be a distribution supported on $[n]$, and let $\mathcal{I} = (I_1, I_2, \ldots, I_r)$ be
  an interval partition of $\mu$ such that $\sum_{j \in [r]: d(\restrict{\mu}{I_j},
  \mathcal{U}_{I_j}) \ge \gamma} \mu(I_j) \le \eta$.  Then $d(\mu, \mu_\mathcal{I}) \le \gamma +
  2\eta$.
  \label{lem:flat-close}
\end{lemma}
\begin{proof}
   We split the sum $d(\mu, \mu_\mathcal{I})$ into parts, one for $I_j$ such that
   $d(\restrict{\mu}{I_j}, \mathcal{U}_{I_j}) \le \gamma$, and one for the remaining intervals.

For $I_j$s such that $d(\restrict{\mu}{I_j}, \mathcal{U}_{I_j}) \le \gamma$, we have
\begin{align}
  \sum_{i \in I_j}\left| \mu(i) - \frac{\mu(I_j)}{|I_j|} \right| &= \sum_{i \in I_j} \mu(I_j) \left|
  \restrict{\mu}{I_j}(i) - \frac{1}{|I_j|} \right| = \mu(I_j)d(\restrict{\mu}{I_j},
  \mathcal{U}_{I_j}) \le \gamma\mu(I_j).
  \label{eqn:dist-heavy}
\end{align}

For $I_j$ such that $d(\restrict{\mu}{I_j}, \mathcal{U}_{I_j}) > \gamma$, we have
\begin{align}
  \sum_{i \in I_j} \left| \mu(i) - \frac{\mu(I_j)}{|I_j|} \right| &= \sum_{i \in I_j} \mu(I_j)
  \left| \restrict{\mu}{I_j}(i) - \frac{1}{|I_j|} \right| \le 2\mu(I_j)
  \label{eqn:dist-light}
\end{align}
We know that the sum of $\mu(I_j)$ over all $I_j$ such that $d(\restrict{\mu}{I_j},
\mathcal{U}_{I_j}) \ge \gamma$ is at most $\eta$.  Using Equations \ref{eqn:dist-light} and
\ref{eqn:dist-heavy}, and summing up over all the sets $I_j \in \mathcal{I}$, the lemma follows.
\end{proof}

The good thing about a flattening (for an interval partition of small length) is that it can be efficiently learned. For this we first make a technical definition and note some trivial observations:

\begin{definition}[coarsening]
 Given $\mu$ and $\mathcal{I}$, where $|\mathcal{I}|=\ell$, we define the {\em coarsening} of $\mu$ according to $\mathcal{I}$ the distribution $\hat{\mu}_{\mathcal{I}}$ over $[\ell]$ as by $\hat{\mu}_{\mathcal{I}}(j)=\mu(I_j)$ for all $j\in [\ell]$.
\end{definition}

\begin{observation}\label{obs:easycoarse}
 Given a distribution $\hat{\mu}_{\mathcal{I}}$ over $[\ell]$, define $\mu_{\mathcal{I}}$ over $[n]$ by $\mu(i)=\hat{\mu}_{\mathcal{I}}(j_i)/|I_{j_i}|$, where $j_i$ is the index satisfying $i\in I_{j_i}$. This is a distribution, and for any two distributions $\hat{\mu}_{\mathcal{I}}$ and $\hat{\chi}_{\mathcal{I}}$ we have $d(\mu_{\mathcal{I}},\chi_{\mathcal{I}})=d(\hat{\mu}_{\mathcal{I}},\hat{\chi}_{\mathcal{I}})$. Moreover, if $\hat{\mu}_{\mathcal{I}}$ is a coarsening of a distribution $\mu$ over $[n]$, then $\mu_{\mathcal{I}}$ is the respective flattening of $\mu$.
\end{observation}
\begin{proof}
 All of this follows immediately from the definitions.
\end{proof}

The following lemma shows how learning can be achieved. We will ultimately use this in conjunction with Lemma \ref{lem:flat-close} as a means to learn a whole distribution through its flattening.

\begin{lemma}
  Given a distribution $\mu$ supported over $[n]$ and an interval partition $\mathcal{I} = (I_1,
  I_2, \ldots, I_\ell)$, using $\tfrac{2(\ell + \log(2/\delta))}{\epsilon^2}$, we can obtain an
  explicit distribution $\mu'_\mathcal{I}$, supported over $[n]$, such that, with probability at
  least $1 - \delta$, $d(\mu_\mathcal{I}, \mu'_\mathcal{I}) \le \epsilon$.
  \label{lem:learn-flat}
\end{lemma}
\begin{proof}
  First, note that an unconditional sample from $\hat{\mu}_\mathcal{I}$ can be simulated using one unconditional sample from $\mu$. To obtain it, take the index $i$ sampled from $\mu$, and set $j$ to be the index for which $i\in I_j$. Using Lemma \ref{lem:learn-dist}, we can now obtain a distribution $\hat{\mu}'_\mathcal{I}$, supported over $[\ell]$, such that with
  probability at least $1-\delta$, $d(\hat{\mu}_\mathcal{I}, \hat{\mu}'_\mathcal{I}) \le \epsilon$. To finish, we construct and output $\mu'_{\mathcal{I}}$ as per Observation \ref{obs:easycoarse}.
\end{proof}


\section{Weakly tolerant interval uniformity tests}\label{sec:intervaltests}

To unify our treatment of learning and testing with respect to $L$-decomposable properties to all three models (unconditional, adaptive-condition and non-adaptive-conditional), we first define what it means to test a distribution $\mu$ for uniformity over an interval $I\subseteq [n]$. The following definition is technical in nature, but it is what we need to be used as a building block for our learning and testing algorithms.

\begin{definition}[weakly tolerant interval tester]
 A {\em weakly tolerant interval tester} is an algorithm $\mathbb{T}$ that takes as input a distribution $\mu$ over $[n]$, an interval $I\subseteq [n]$, a maximum size parameter $m$, a minimum weight parameter $\gamma$, an approximation parameter $\epsilon$ and an error parameter $\delta$, and satisfies the following.
 \begin{enumerate}
  \item If $|I|\le m$, $\mu(I)\ge\gamma$, and $\bias(\restrict{\mu}{I}) \le \epsilon/100$, then the algorithm accepts with probability at least $1-\delta$.
  \item If $|I|\le m$, $\mu(I)\ge\gamma$, and $d(\restrict{\mu}{I},\mathcal{U}_I) > \epsilon$, then the algorithm rejects with probability at least $1-\delta$.
 \end{enumerate}
 In all other cases, the algorithm may accept or reject with arbitrary probability.
\end{definition}

For our purposes we will use three weakly tolerant interval testers, one for each model.

First, a tester for uniformity which uses unconditional samples, a
version of which has already appeared implicitly in \cite{GoldreichR00}. We state below
the tester with the best dependence on $n$ and $\epsilon$. We first state it in its original form, where $I$ is the whole of $[n]$, implying that $m=n$ and $\gamma=1$, and $\delta=1/3$.
\begin{lemma}[\cite{Paninski08}]
  For the input $(\mu,[n],n,1,\epsilon,1/3)$, there is a weakly tolerant interval tester
  using $O(\sqrt{n}/\epsilon^2)$ unconditional samples from $\mu$.
  \label{lem:uniformity}
\end{lemma}

The needed adaptation to our purpose is straightforward.
\begin{lemma}\label{lem:uniformity-interval}
 For the input $(\mu,I,m,\gamma,\epsilon,\delta)$, there is a weakly tolerant interval tester
 which uses $O(\sqrt{m}\log(1/\delta)/\gamma\epsilon^2)$ unconditional samples from $\mu$.
\end{lemma}
\begin{proof}
 To adapt the tester of Lemma \ref{lem:uniformity} to the general $m$ and $\gamma$, we just take samples according to $\mu$ and keep from them those samples lie in $I$. This simulates samples from $\restrict{\mu}{I}$, over which we employ the original tester. This gives a tester using $O(\sqrt{m}/\gamma\epsilon^2)$ unconditional samples and providing an error parameter of, say, $\delta=2/5$ (the extra error is due to the probability of not getting enough samples from $I$ even when $\mu(I)\ge\gamma$). To move to a general $\delta$, we repeat this $O(1/\delta)$ times and take the majority vote.
\end{proof}

Next, a tester that uses adaptive conditional samples. For this we use the following tester from \cite{CanonneRS15} (see also \cite{ChakrabortyFGM13}). Its original statement does not have the weakly tolerance (acceptance for small bias) guarantee, but it is easy to see that the proof there works for the stronger assertion. This time we skip the question of how to adapt the original algorithm from $I=[n]$ and $\delta=2/3$ to the general parameters here. This is since $\gamma$ does not matter (due to using adaptive conditional samples), the query complexity is independent of the domain size to begin with, and the move to a general $\delta>0$ is by standard amplification.

\begin{lemma}
  [\cite{CanonneRS15}, see also \cite{ChakrabortyFGM13}]
  For the input $(\mu,I,m,\gamma,\epsilon,\delta)$, there is a weakly tolerant interval tester
  that adaptively takes $\log(1/\delta)\poly(\log(1/\epsilon))/\epsilon^2$ conditional
  samples from $\mu$.
  \label{lem:adaptive-uniformity-interval}
\end{lemma}

Finally, a tester that uses non-adaptive conditional samples. For this to work, it is also very important that the queries do not depend on $I$ as well (but only on $n$ and $\gamma$). We just state here the lemma, the algorithm itself is presented and analyzed in Section \ref{sec:non-adaptive}.

\begin{lemma}
  For the input $(\mu,I,m,\gamma,\epsilon,\delta)$, there is a weakly tolerant interval tester
  that non-adaptively takes $\poly(\log n,1/\epsilon)\log(1/\delta)/\gamma$ conditional
  samples from $\mu$, in a manner independent of the interval $I$.
  \label{lem:nonadaptive-uniformity-interval}
\end{lemma}


\section{Assessing an interval partition}\label{sec:assessment}

Through either Lemma \ref{lem:pulled-partition} or Lemma \ref{lem:pulled-rel-partition} we know how to construct a fine partition, and then through either Lemma \ref{lem:small-wt-non-uniform} or Lemma \ref{lem:small-wt-non-uniform-rel} respectively we know that if $\mu$ is decomposable, then most of the weight is concentrated on intervals with a small bias. However, eventually we would like a test that works for decomposable and non-decomposable distributions alike. For this we need a way to asses an interval partition as to whether it is indeed suitable for learning a distribution. This is done through a weighted sampling of intervals, for which we employ a weakly tolerant tester, The following is the formal description, given as Algorithm \ref{alg:assess}.

\begin{algorithm}
  [htpb]
  \caption{Assessing a partition}
  \label{alg:assess}

  \KwInput{A distribution $\mu$ supported over $[n]$, parameters $c,r$, an interval partition $\mathcal{I}$ satisfying $|\mathcal{I}|\le r$, parameters $\epsilon, \delta > 0$, a weakly tolerant interval uniformity tester $\mathbb{T}$ taking input values $(\mu,I,m,\gamma,\epsilon,\delta)$.}

\For {$s = 20\log(1/\delta)/\epsilon$ times}
{
	Take an unconditional sample from $\mu$ and let  $I \in \mathcal{I}$ be the interval that contains it
	\label{step:assess-single-sample}
	
    Use the tester $\mathbb{T}$ with input values $(\mu,I,n/c,\epsilon/r,\epsilon,\delta/2s)$
    \label{step:assess-single-test}

    \lIf {test rejects}{add $I$ to $\mathcal{B}$}
    \label{step:asses-classify}
}
  \lIf*{$|\mathcal{B}| > 4\epsilon s$}{$\reject$} \lElse*{$\accept$}
  \label{step:assess-threshhold}

\end{algorithm}

To analyze it, first, for a fine interval partition, we bound the total weight of intervals where the weakly tolerant tester is not guaranteed a small error probability; recall that $\mathbb{T}$ as used in Step \ref{step:assess-single-test} guarantees a correct output only for an interval $I$ satisfying $\mu(I)\geq \epsilon/r$ and $|I|\leq n/r$.

\begin{observation}\label{obs:dubious-intervals}
 Define $\mathcal{N}_{\mathcal{I}}=\{I\in\mathcal{I}:|I|>n/r~\mathrm{or}~\mu(I)<\epsilon/r\}$. If $\mathcal{I}$ is $(\eta,\gamma)$-fine, where $c\eta+\gamma\leq\epsilon$, then $\mu(\bigcup_{I\in\mathcal{N}_{\mathcal{I}}}I)\le 2\epsilon$.
\end{observation}
\begin{proof}
 Intervals in $\mathcal{N}_{\mathcal{I}}$ must fall into at least one of the following categories.
 \begin{itemize}
  \item Intervals in $\mathcal{H}_{\mathcal{I}}$, whose total weight is bounded by $\gamma$ by the definition of a fine partition.
  \item Intervals whose weight is less than $\epsilon/r$. Since there are at most $r$ such intervals (since $|\mathcal{I}|\leq r$) their total weight is bounded by $\epsilon$.
  \item Intervals whose size is more than $n/c$ and are not in $\mathcal{H}_{\mathcal{I}}$. Every such interval is of weight bounded by $\eta$ (by the definition of a fine partition) and clearly there are no more than $c$ of those, giving a total weight of $c\eta$.
 \end{itemize}
Summing these up concludes the proof.
\end{proof}

The following ``completeness'' lemma states that the typical case for a fine partition of a decomposable distribution, i.e. the case where most intervals exhibit a small bias, is correctly detected.

\begin{lemma}\label{lem:assess-complete}
 Suppose that $\mathcal{I}$ is $(\eta,\gamma)$-fine, where $c\eta+\gamma\leq\epsilon$.
 Define $\mathcal{G}_{\mathcal{I}}=\{i:\mathcal{I}:\bias(\restrict{\mu}{I})\le \epsilon/100\}$. If $\mu(\bigcup_{I\in\mathcal{G}_{\mathcal{I}}})\ge 1-\epsilon$, then Algorithm \ref{alg:assess} accepts with probability at least $1-\delta$.
\end{lemma}
\begin{proof}
 Note by Observation \ref{obs:dubious-intervals} that the total weight of $\mathcal{G}_{\mathcal{I}}\setminus\mathcal{N}_{\mathcal{I}}$ is at least $1-3\epsilon$. By the Chernoff bound of Lemma \ref{lem:chernoff}, with probability at least $1-\delta/2$ all but at most $4\epsilon s$ of the intervals drawn in Step \ref{step:assess-single-sample} fall into this set.
 
 Finally, note that if $I$ as drawn in Step \ref{step:assess-single-sample} belongs to this set, then with probability at least $1-\delta/2s$ the invocation of $\mathbb{T}$ in Step \ref{step:assess-single-test} will accept it, so by a union bound with probability at least $1-\delta/2$ all sampled intervals from this set will be accepted. All events occur together and make the algorithm accept with probability at least $1-\delta$, concluding the proof.
\end{proof}

The following ``soundness'' lemma states that if too much weight is concentrated on intervals where $\mu$ is far from uniform in the $\ell_1$ distance, then the algorithm rejects. Later we will show that this is the only situation where $\mu$ cannot be easily learned through its flattening according to $\mathcal{I}$.

\begin{lemma}\label{lem:assess-sound}
 Suppose that $\mathcal{I}$ is $(\eta,\gamma)$-fine, where $c\eta+\gamma\leq\epsilon$.
 Define $\mathcal{F}_{\mathcal{I}}=\{i:\mathcal{I}:d(\restrict{\mu}{I},\mathcal{U}_I) > \epsilon\}$. If $\mu(\bigcup_{I\in\mathcal{F}_{\mathcal{I}}})\ge 7\epsilon$, then Algorithm \ref{alg:assess} rejects with probability at least $1-\delta$.
\end{lemma}
\begin{proof}
 Note by Observation \ref{obs:dubious-intervals} that the total weight of $\mathcal{F}_{\mathcal{I}}\setminus\mathcal{N}_{\mathcal{I}}$ is at least $5\epsilon$. By the Chernoff bound of Lemma \ref{lem:chernoff}, with probability at least $1-\delta/2$ at least $4\epsilon s$ of the intervals drawn in Step \ref{step:assess-single-sample} fall into this set.
 
 Finally, note that if $I$ as drawn in Step \ref{step:assess-single-sample} belongs to this set, then with with probability at least $1-\delta/2s$ the invocation of $\mathbb{T}$ in Step \ref{step:assess-single-test} will reject it, so by a union bound with probability at least $1-\delta/2$ all sampled intervals from this set will be rejected. All events occur together and make the algorithm reject with probability at least $1-\delta$, concluding the proof.
\end{proof}

Finally, we present the query complexity of the algorithm. It is presented as generally quadratic in $\log(1/\delta)$, but this can be made linear easily by first using the algorithm with $\delta=1/3$, and then repeating it $O(1/\delta)$ times and taking the majority vote. When we use this lemma later on, both $r$ and $c$ will be linear in the decomposability parameter $L$ for a fixed $\epsilon$, and $\delta$ will be a fixed constant.

\begin{lemma}\label{lem:assess-complexity}
 Algorithm \ref{alg:assess} requires $O(q\log(1/\delta)/\epsilon)$ many samples, where $q=q(n/c,\epsilon/r,\epsilon,\delta/2s)$ is the number of samples that the invocation of $\mathbb{T}$ in Step \ref{step:assess-single-test} requires.
 
 In particular, Algorithm \ref{alg:assess} can be implemented either as an unconditional sampling algorithm taking $r\sqrt{n/c}\log^2(1/\delta)/\poly(\epsilon)$ many samples, an adaptive conditional sampling algorithm taking $r\log^2(1/\delta)/\poly(\epsilon)$ many samples, or a non-adaptive conditional sampling algorithm taking $r\log^2(1/\delta)\poly(\log n,1/\epsilon)$ many samples.
\end{lemma}
\begin{proof}
 A single (unconditional) sample is taken each time Step \ref{step:assess-single-sample} is reached, and all other samples are taken by the invocation of $\mathbb{T}$ in Step \ref{step:assess-single-test}. This makes the total number of samples to be $s(q+1)=O(q\log(1/\delta)/\epsilon)$.
 
 The bound for each individual sampling model follows by plugging in Lemma \ref{lem:uniformity-interval}, Lemma \ref{lem:adaptive-uniformity-interval} and Lemma \ref{lem:nonadaptive-uniformity-interval} respectively. For the last one it is important that the tester makes its queries completely independently of $I$, as otherwise the algorithm would not have been non-adaptive.
\end{proof}


\section{Learning and testing decomposable distributions and properties}\label{sec:learn}

Here we finally put things together to produce a learning algorithm for $L$-decomposable distribution. This algorithm is not only guaranteed to learn with high probability a distribution that is decomposable, but is also guaranteed with high probability to not produce a wrong output for any distribution (though it may plainly reject a distribution that is not decomposable).

This is presented in Algorithm \ref{alg:learn}. We present it with a fixed error probability $2/3$ because this is what we use later on, but it is not hard to move to a general $\delta$.

\begin{algorithm}
   [htpb]
   \caption{Learning an $L$-decomposable distribution}
   \label{alg:learn}
   
   \KwInput{Distribution $\mu$ supported over $[n]$, parameters $L$ (decomposability), $\epsilon>0$ (accuracy), a weakly tolerant interval uniformity tester $\mathbb{T}$ taking input values $(\mu,I,m,\gamma,\epsilon,\delta)$}

   Use Algorithm \ref{alg:pulled-rel-partition} with input values $(\mu,\epsilon/2000L,\epsilon/2000,1/9)$ to obtain a partition $\mathcal{I}$ with $|\mathcal{I}|\leq r=10^5L\log(1/\epsilon)/\epsilon$
   \label{step:learn-pull}
   
   Use Algorithm \ref{alg:assess} with input values $(\mu,L,r,\mathcal{I},\epsilon/20,1/9,\mathbb{T})$
   \label{step:learn-assess}
   
   \lIf{Algorithm \ref{alg:assess} rejected}{$\reject$}
   \label{step:learn-prereject}
   
   Use Lemma \ref{lem:learn-flat} with values $(\mu,\mathcal{I},\epsilon/10,1/9)$ to obtain $\mu'_{\mathcal{I}}$
   \label{step:learn-learn}
      
 \Return{$\mu'_{\mathcal{I}}$}
 
\end{algorithm}

First we show completeness, that the algorithm will be successful for decomposable distributions.

\begin{lemma}\label{lem:learn-complete}
 If $\mu$ is $(\epsilon/2000,L)$-decomposable, then with probability at least $2/3$ Algorithm \ref{alg:learn} produces a distribution $\mu'$ so that $d(\mu,\mu')\leq \epsilon$.
\end{lemma}
\begin{proof}
 By Lemma \ref{lem:pulled-rel-partition}, with probability at least $8/9$ the partition $\mathcal{I}$ is $(\epsilon/2000L,\epsilon/2000)$-fine, which means by Lemma \ref{lem:small-wt-non-uniform-rel} that $\sum_{j \in [r]: \bias(\restrict{\mu}{I'_j}) > \epsilon/2000} \mu(I'_j) \le 3\epsilon/2000$. When this occurs, by Lemma \ref{lem:assess-complete} with probability at least $8/9$ Algorithm \ref{alg:assess} will accept and so the algorithm will move past Step \ref{step:learn-prereject}. In this situation, in particular by Lemma \ref{lem:flat-close} we have that $d(\mu_{\mathcal{I}},\mu)\le 15\epsilon/20$ (in fact this can be bounded much smaller here), and with probability at least $8/9$ (by Lemma \ref{lem:learn-flat}) Step \ref{step:learn-learn} provides a distribution that is $\epsilon/10$-close to $\mu_{\mathcal{I}}$ and hence $\epsilon$-close to $\mu$.
\end{proof}

Next we show soundness, that the algorithm will with high probability not mislead about the distribution, whether it is decomposable or not.

\begin{lemma}\label{lem:learn-sound}
 For any $\mu$, the probability that Algorithm \ref{alg:learn} produces (without rejecting) a distribution $\mu'$ for which $d(\mu,\mu')>\epsilon$ is bounded by $\delta$.
\end{lemma}
\begin{proof}
 Consider the interval partition $\mathcal{I}$. By Lemma \ref{lem:pulled-rel-partition}, with probability at least $8/9$ it is $(\epsilon/2000L,\epsilon/2000)$-fine. When this happens, if $\mathcal{I}$ is such that $\sum_{j:d(\restrict{\mu}{I_j},\mathcal{U}_{I_j})}\mu(I_j)>7\epsilon/20$, then by Lemma \ref{lem:assess-sound} with probability at least $8/9$ the algorithm will reject in Step \ref{step:learn-prereject}, and we are done (recall that here a rejection is an allowable outcome).
 
 On the other hand, if $\mathcal{I}$ is such that $\sum_{j:d(\restrict{\mu}{I_j},\mathcal{U}_{I_j})}\mu(I_j)\le 7\epsilon/20$, then by Lemma \ref{lem:flat-close} we have that $d(\mu_{\mathcal{I}},\mu)\le 15\epsilon/20$, and with probability at least $8/9$ (by Lemma \ref{lem:learn-flat}) Step \ref{step:learn-learn} provides a distribution that is $\epsilon/10$-close to $\mu_{\mathcal{I}}$ and hence $\epsilon$-close to $\mu$, which is also an allowable outcome.
\end{proof}

And finally, we plug in the sample complexity bounds.

\begin{lemma}\label{lem:learn-complexity}
 Algorithm \ref{alg:learn} requires $O(L\log(1/\epsilon)/\epsilon+q/\epsilon+L\log(1/\epsilon)/\epsilon^3)$ many samples, where the value $q=q(n/L,\epsilon^2/10^5L\log(1/\epsilon),\epsilon/20,2000/\epsilon)$ is a bound on the number of samples that each invocation of $\mathbb{T}$ inside Algorithm \ref{alg:assess} requires.
 
 In particular, Algorithm \ref{alg:learn} can be implemented either as an unconditional sampling algorithm taking $\sqrt{nL}/\poly(\epsilon)$ many samples, an adaptive conditional sampling algorithm taking $L/\poly(\epsilon)$ many samples, or a non-adaptive conditional sampling algorithm taking $L\poly(\log n,1/\epsilon)$ many samples.
\end{lemma}
\begin{proof}
 The three summands in the general expression follow respectively from the sample complexity calculations of Lemma \ref{lem:pulled-rel-partition} for Step \ref{step:learn-pull}, Lemma \ref{lem:assess-complexity} for Step \ref{step:learn-assess}, and Lemma \ref{lem:learn-flat} for Step \ref{step:learn-learn} respectively. Also note that all samples outside Step \ref{step:learn-assess} are unconditional.
 
 The bound for each individual sampling model follows from the respective bound stated in Lemma \ref{lem:assess-complexity}.
\end{proof}

Let us now summarize the above as a theorem.

\begin{theorem}\label{thm:learn}
 Algorithm \ref{alg:learn} is capable of learning an $(\epsilon/2000,L)$-decomposable distribution, giving with probability at least $2/3$ a distribution that is $epsilon$-close to it, such that for no distribution will it give as output a distribution $\epsilon$-far from it with probability more than $1/3$.
 
 It can be implemented either as an unconditional sampling algorithm taking $\sqrt{nL}/\poly(\epsilon)$ many samples, an adaptive conditional sampling algorithm taking $L/\poly(\epsilon)$ many samples, or a non-adaptive conditional sampling algorithm taking $L\poly(\log n,1/\epsilon)$ many samples.
\end{theorem}
\begin{proof}
 This follows from Lemmas \ref{lem:learn-complete}, \ref{lem:learn-sound} and \ref{lem:learn-complexity} respectively.
\end{proof}

Let us now move to the immediate application of the above for testing decomposable properties. The algorithm achieving this is summarized as Algorithm \ref{alg:test}

\begin{algorithm}
  [htpb]
  \caption{Testing $L$-decomposable properties.}
  \label{alg:test}

  \KwInput{Distribution $\mu$ supported over $[n]$, function $L:(0,1] \times \mathbb N \to \mathbb N$ (decomposability), parameter $\epsilon > 0$ (accuracy), an $L$-decomposable property $\mathcal{C}$ of distributions, a weakly tolerant interval uniformity tester $\mathbb{T}$ taking input values $(\mu,I,m,\gamma,\epsilon,\delta)$.}

    Use Algorithm \ref{alg:learn} with input values $(\mu,L(\epsilon/4000,n),\epsilon/2,\mathbb{T})$ to obtain $\mu'$
    \label{step:test-learn}

    \lIf*{Algorithm \ref{alg:learn} accepted and $\mu'$ is $\epsilon/2$-close to $\mathcal{C}$} {$\accept$} \lElse*{$\reject$}
    \label{step:test-answer}
\end{algorithm}

\begin{theorem}\label{thm:test}
 Algorithm \ref{alg:test} is a test (with error probability $1/3$) for the $L$-decomposable property $\mathcal{C}$.
 For $L=L(\epsilon/4000,n)$,
 It can be implemented either as an unconditional sampling algorithm taking $\sqrt{nL}/\poly(\epsilon)$ many samples, an adaptive conditional sampling algorithm taking $L/\poly(\epsilon)$ many samples, or a non-adaptive conditional sampling algorithm taking $L\poly(\log n,1/\epsilon)$ many samples.
\end{theorem}
\begin{proof}
 The number and the nature of the samples are determined fully by the application of Algorithm \ref{alg:learn} in Step \ref{step:test-learn}, and are thus the same as in Theorem \ref{thm:learn}. Also by this theorem, for a distribution $\mu\in\mathcal{C}$, with probability at least $2/3$ an $\epsilon/2$-close distribution $\mu'$ will be produced, and so it will be accepted in Step \ref{step:test-answer}.
 
 Finally, if $\mu$ is $\epsilon$-far from $\mathcal{C}$, then with probability at least $2/3$ Step \ref{step:test-learn} will either produce a rejection, or again produce $\mu'$ that is $\epsilon/2$-close to $\mu$. In the latter case, $\mu'$ will be $\epsilon/2$-far from $\mathcal{C}$ by the triangle inequality, and so Step \ref{step:test-answer} will reject in either case.
\end{proof}


\section{A weakly tolerant tester for the non-adaptive conditional model}\label{sec:non-adaptive}

Given a distribution $\mu$, supported over $[n]$, and an interval $I \subseteq [n]$ such that
$\mu(I) \ge \gamma$, we give a tester that uses non-adaptive conditional queries to $\mu$ to
distinguish between the cases $\bias(\restrict{\mu}{I}) \le \epsilon/100$ and $d(\restrict{\mu}{I}, \mathcal{U}_I) > \epsilon$, using ideas from
\cite{ChakrabortyFGM13}. A formal description of the test is given as Algorithm
\ref{alg:weakly-tolerant}. It is formulated here with error probability $\delta=1/3$. Lemma \ref{lem:nonadaptive-uniformity-interval} is obtained from this the usual way, by repeating the algorithm $O(1/\delta)$ times and taking the majority vote.

\begin{algorithm}
  [htpb]
  \caption{Non-adaptive weakly tolerant uniformity tester}
  \label{alg:weakly-tolerant}
  \KwInput{Distribution $\mu$ supported over $[n]$, interval $I \subseteq [n]$,
  weight bound $\gamma$, accuracy $\epsilon > 0$.}

  Sample $t = \frac{4(\log^{10} n + 3)}{\epsilon^2 \gamma}$ elements
  from $\mu$. 
  \label{step:uncond-samples}

  \For {$k \in \{0, \ldots, \log n\}$}
  {
    Set $p_k = 2^{-k}$.

    Choose a set $U_k \subseteq [n]$, where each $i \in [n]$ is in $U_k$ with probability $p_k$,
    independently of other elements in $[n]$.
  }

  \If {$|I| \le \log^{10} n$}
  {
    Use Lemma \ref{lem:brute-force}, using the $t$ unconditional samples from $\mu$, to construct a
    distribution $\mu'$
    \label{step:bf-test}

    \lIf* {$d(\mu', \mathcal{U}_I) \le \epsilon/2$} {\accept,} \lElse* {\reject.}
    \label{step:small-set}
  }
  \Else
  {
    \For {$U_k$ such that $k \le \log\left( \frac{|I|}{2\log^8 n} \right)$, and  $|I \cap U_k| \ge
  \log^8 n$}
  {\label{step:loop-collision}

      Sample $\log^3 n$ elements from $\restrict{\mu}{U_k}$.
      \label{step:sample-collision}

      \lIf {the same element from $I \cap U_k$ has been sampled twice} {$\reject$.}
	\label{step:wt-collision}
    }

    Choose an index $k$ such that $\frac{2}{3}\log^8 n \le |I|p_k < \frac{4}{3}\log^8 n$.
    \label{step:choose-k}

    Sample $m_k = \frac{C\log^{16} n \log(3\log n)}{\epsilon^2 \gamma}$
    elements from $\restrict{\mu}{U_k}$, for a large constant $C$.
    \label{step:sample-close}

    \If {$|I \cap U_k| > 2|I|p_k$ or the number of samples in $I \cap U_k$ is less than
    $\gamma m_k/40$}
    {
       \label{step:wt-too-few}

      \reject.
    } 
    \Else
    {

      Use Lemma \ref{lem:learn-infty} with the samples received from $I \cap U_k$, to construct
      $\mu'$, supported on $I\cap U_k$, such that $\lVert \mu' - \restrict{\mu}{I\cap
      U_k} \rVert_\infty \le \frac{\epsilon}{80|I \cap U_k|}$ with probability at least $9/10$.
      \label{step:learn-infty}

      \lIf* {$\lVert \mu' - \mathcal{U}_{I \cap U_k} \rVert_\infty \le
      \frac{3\epsilon}{80|I \cap U_k|}$} {$\accept$,} \lElse* {$\reject$.}
      \label{step:final-answer}
    }
  }
\end{algorithm}

We first make the observation that makes Algorithm \ref{alg:weakly-tolerant} suitable for a
non-adaptive setting.
\begin{observation}
  Algorithm \ref{alg:weakly-tolerant} can be implemented using only non-adaptive conditional queries
  to the distribution $\mu$, that are chosen independently of $I$.
\end{observation}
\begin{proof}
  First, note that the algorithm samples elements from $\mu$ at three places. Initially, it samples
  unconditionally from $\mu$ in Step \ref{step:uncond-samples},  and then it performs conditional
  samples from the sets $U_k$ in Steps \ref{step:sample-collision} and \ref{step:sample-close}. In
  Steps \ref{step:sample-collision} and \ref{step:sample-close}, the samples are conditioned on sets
  $U_k$, where $k$ depends on $I$. However, observe that we can sample from all sets $U_k$, for all
  $0 \le k \le \log n$, at the beginning,  and then use just the samples from the appropriate $U_k$
  at Steps \ref{step:sample-collision} and \ref{step:sample-close}.  This only increases the bound
  on the number of samples by a factor of $\log n$. Thus we have only non-adaptive queries, all of
  which are made at the start of the algorithm, independently of $I$.
\end{proof}

The following lemma is used in Step \ref{step:bf-test} of our algorithm.

\begin{lemma}
  Let $\mu$ be a distribution supported over $[n]$ and $I \subseteq [n]$ be an interval such that
  $\mu(I) \ge \gamma$. Using $t = \frac{4(|I| + \log(2/\delta))}{\epsilon^2 \gamma} $ unconditional
  queries to $\mu$, we can construct a distribution $\mu'$ over $I$ such that, with probability at
  least $1 - \delta$, $d( \restrict{\mu}{I}, \mu') \le \epsilon$ (in other cases $\mu'$ may be arbitrary).
  \label{lem:brute-force}
\end{lemma}
\begin{proof}
  Take $t = \frac{4(|I| + \log(2/\delta))}{\epsilon^2 \gamma}$ unconditional samples. Let $t_I$ be
  the number of samples that belong to $I$. Then, $E[t_I] = t\mu(I) \ge t\gamma$.  Therefore, by
  Hoeffding bounds, with probability at least $1 - \exp(-t\mu(I)/4)$, $ t_I \ge t\mu(I)/2 \ge
  t\gamma/2$.  

  The $t_I$ samples are distributed according to $\restrict{\mu}{I}$. By the choice of $t$, with
  probability at least $1 - \delta/2$, $t_I \ge 2(|I| + \log(2/\delta))/\epsilon^2$. Therefore, by
  Lemma \ref{lem:learn-dist}, we can obtain a distribution $\mu'$, supported over $I$, such that
  with probability at least $1 - \delta$, $d(\restrict{\mu}{I}, \mu') \le \epsilon$.
  
  If we did not obtain sufficiently many samples (either because $\mu(I)<\gamma$ or due to a low probability event) then we just output an arbitrary distribution supported on $I$.
\end{proof}

\begin{lemma}[Completeness]
  If $\mu(I)\ge\gamma$ and $\bias(\restrict{\mu}{I}) \le \epsilon/100$, then
  Algorithm \ref{alg:weakly-tolerant} accepts with probability at least $2/3$.
  \label{lem:wt-complete}
\end{lemma}
\begin{proof}
  First note that if $|I| \le \log^{10} n$, then we use Lemma \ref{lem:brute-force} to test the
  distance of $\restrict{\mu}{I}$ to uniform with probability at least $9/10$ in Step
  \ref{step:small-set}. For the remaining part of the proof, we will assume that $|I| > \log^{10}
  n$.

  For a set $U_k$ chosen by the algorithm, and any $i \in I \cap U_k$, the probability that it is
  sampled twice in Step \ref{step:wt-collision} is at most $\binom{\log^3 n}{2}\left(
  \tfrac{\mu(i)}{\mu(U_k)} \right)^2$. Since $\mu(U_k) \ge \mu(I \cap U_k)$, the probability of
  sampling twice in Step \ref{step:wt-collision} is at most $\binom{\log^3 n}{2}\left(
  \tfrac{\mu(i)}{\mu(I \cap U_k)} \right)^2$. By Observation \ref{obs:norms} $\bias(\restrict{\mu}{I}) \le \epsilon/100$ implies $\lVert \restrict{\mu}{I} - \mathcal{U}_I
  \rVert_\infty \le \frac{\epsilon}{100|I|}$, so we have
  \begin{align}
    \frac{\mu(I)}{|I|}\left( 1 - \frac{\epsilon}{100} \right) \le \mu(i) \le \frac{\mu(I)}{|I|}
    \left( 1 + \frac{\epsilon}{100} \right).
    \label{eqn:individual-prob}
  \end{align}
   From Equation \ref{eqn:individual-prob} we get the following for all $U_k$.
  \begin{align}
    \frac{|I\cap U_k| \mu(I)}{|I|} \left( 1 - \frac{\epsilon}{100} \right) \le \mu(I \cap U_k) \le
    \frac{|I\cap U_k| \mu(I)}{|I|} \left( 1 + \frac{\epsilon}{100} \right).
    \label{eqn:total-probability}
  \end{align}
   Therefore, the probability that the algorithm
  samples the same element in $I \cap U_k$ at Step \ref{step:wt-collision} twice is bounded as
  follows.
  \begin{align*}
    \sum_{i \in I \cap U_k} \binom{\log^3 n}{2} \left( \frac{\mu(i)}{\mu(I \cap U_k)} \right)^2 &
    \le |I \cap U_k| \binom{\log^3 n}{2} \frac{\max_{i \in I\cap U_k}\mu(i)^2}{\mu(I \cap U_k)^2}\\
    &\le \frac{1}{|I \cap U_k|} \binom{\log^3 n}{2} \left( \frac{1 + \epsilon/100}{1 - \epsilon/100}
    \right)^2
  \end{align*}

  Since $|I \cap U_k| \ge \log^8 n$ for the $k$ chosen in Step \ref{step:loop-collision}, we can
  bound the sum as follows.
  \begin{align*}
      \sum_{i \in I \cap U_k} \binom{\log^3 n}{2} \left( \frac{\mu(i)}{\mu(I \cap U_k)} \right)^2
      \le \frac{1}{\log^2 n} \left( \frac{1 + \epsilon/100}{1 - \epsilon/100} \right)^2.
  \end{align*}
  Therefore, with probability at least $1 - o(1)$, the algorithm does not reject at Step
  \ref{step:wt-collision}.

  To show that the algorithm accepts with probability at least $2/3$ in Step
  \ref{step:final-answer}, we proceed as follows.  Combining Equations \ref{eqn:individual-prob} and
  \ref{eqn:total-probability}, we get the following.
  \begin{align*}
    \frac{1}{|I\cap U_k|} \left( \frac{1 - \epsilon/100}{1 + \epsilon/100} \right) \le
    \restrict{\mu}{I\cap U_k}(i) \le \frac{1}{|I \cap U_k|}\left( \frac{1 + \epsilon/100}{1 -
    \epsilon/100} \right)
  \end{align*}
  From this it follows that $\lVert \restrict{\mu}{I \cap U_k} - \mathcal{U}_{I \cap U_k}
  \rVert_\infty \le \frac{\epsilon}{40|I \cap U_k|}$. 
 
  We now argue that in this case, the test does not reject at Step \ref{step:wt-too-few}, for the
  $k$ chosen in Step \ref{step:choose-k}. Observe that $E[\mu(I \cap U_k)] 
   \ge p_k \gamma$.
  Also, the expected size of the set $I \cap U_k$ is $p_k|I|$. Since the $k$ chosen in Step
  \ref{step:choose-k} is such that $|I|p_k \ge \tfrac{2}{3}\log^8 n$, with probability at least $1 -
  \exp(-O(\log^8 n))$, $ p_k |I| /2 \le |I \cap U_k| \le 2p_k |I|$ (and in particular Step
  \ref{step:wt-too-few} does not reject). Therefore from Equation \ref{eqn:total-probability}, we
  get that, with probability at least $1 - \exp(-O(\log^8 n))$, $\mu(I \cap U_k) \ge p_k \gamma/3$.
  Since $E[\mu(U_k)] = p_k$, we can conclude using Markov's inequality that, with probability at
  least $9/10$, $\mu(U_k) \le 10p_k$. The expected number of samples from $I \cap U_k$ among the
  $m_k$ samples used in Step $17$ is $m_k \mu(I \cap U_k)/\mu(U_k)$. Therefore, with probability at
  least $9/10$, the expected number of samples from $I \cap U_k$ among the $m_k$ samples is at least
  $m_k \gamma/30$. Therefore, with probability, at least $9/10 - o(1)$, at least $m_k \gamma/40$
  elements of $I \cap U_k$ are sampled, and the tester does not reject at Step
  \ref{step:wt-too-few}. The indexes that are sampled in Step \ref{step:sample-close} that lie in $I
  \cap U_k$ are distributed according to $\restrict{\mu}{I \cap U_k}$ and we know that $|I \cap U_k|
  \le 2|I|p_k \le \tfrac{8}{3}\log^8 n$. Therefore, with probability at least $9/10$, we get a
  distribution $\mu'$ such that $\lVert \mu' - \restrict{\mu}{I \cap U_k} \rVert_\infty \le
  \tfrac{\epsilon}{80|I \cap U_k|}$ in Step \ref{step:learn-infty}.

  Therefore, the test correctly accepts in Step \ref{step:final-answer} for the $k$ chosen in Step
  \ref{step:choose-k}.
\end{proof}

Now we prove the soundness of the tester mentioned above. First we state a lemma from Chakraborty et
al \cite{ChakrabortyFGM13}.
\begin{lemma}[\cite{ChakrabortyFGM13}, adapted for intervals]
  Let $\mu$ be a distribution, and $I\subseteq [n]$ be an interval such that
  $d(\restrict{\mu}{I}, \mathcal{U}_I) \ge \epsilon$. Then the following two conditions hold.
  \begin{enumerate}
   \item There exists a set $B_1 = \left\{ i \in I \mid \restrict{\mu}{I}(i) < \frac{1 +
     \epsilon/3}{|I|} \right\}$ such that $|B_1| \ge \epsilon |I|/2$.
   \item  There exists an index $j \in \left\{3, \ldots, \frac{\log |I|}{\log(1 + \epsilon/3)}
   \right\}$, and a set $B_j$ of cardinality at least $\frac{\epsilon^2 |I|}{96\left( 1 + \epsilon/3
   \right)^j \log |I|}$, such that $\frac{(1 + \epsilon/3)^{j - 1}}{|I|} \le \restrict{\mu}{I}(i) <
   \frac{(1 + \epsilon/3)^j}{|I|}$ for all $i \in B_j$.
  \end{enumerate}                                
  \label{lem:bucketing}                          
\end{lemma}                                      
                                                 
Now we analyze the case where $d(\restrict{\mu}{I }, \mathcal{U}_{I}) > \epsilon$.
                                                 
\begin{lemma}[Soundness]
  Let $\mu$ be a distribution supported on $[n]$, and let $I \subseteq [n]$ be an interval such that
  $\mu(I) \ge \gamma$. If $d(\restrict{\mu}{I}, \mathcal{U}_I) \ge \epsilon$, then Algorithm
  \ref{alg:weakly-tolerant} rejects with probability at least $2/3$.
  \label{lem:wt-sound}
\end{lemma}

\begin{proof}
  Observe that when $|I| \le \log^{10} n$, the algorithm rejects with probability at least $9/10$ in
  Step \ref{step:small-set}. For the remainder of the proof, we will assume that $|I| > \log^{10}
  n$.  We analyze two cases according to the value of $j$ given by Lemma \ref{lem:bucketing}.

\begin{enumerate}
  \item Suppose that $j > 2$ is such that $\left| B_j \right| \ge \frac{\epsilon^2 |I|}{96\left( 1 +
    \epsilon/3 \right)^j \log |I|}$, and $(1 + \epsilon/3)^j \le \log^6 n$.  The expected number of
    elements from this set that is chosen in $U_k$ is at least $\frac{\epsilon^2 |I| p_k}{96(1 +
    \epsilon/3)^j \log|I|}$. 
    For the choice of $k$ made in Step \ref{step:choose-k}, we have $|I|p_k \ge \tfrac{2}{3}\log^8
    n$. The probability that no index from $B_j$ is chosen in $U_k$ is $(1 - p_k)^{|B_j|}$ which is
    at most $( 1 - \tfrac{2 \log^8 n}{3|I|})^{\epsilon^2 |I|/(1 + \epsilon)^j \log|I|}$. Since
    $(1 + \epsilon/3)^j \le \log^6 n$, this is at most $\exp\left( - \tfrac{\epsilon^2
    \log n}{144} \right)$. Therefore, with probability $1 - o(1)$, at least one element $i$ is chosen
    from $B_j$.
    
    Since $|B_1| \ge \epsilon |I|/2$, the probability that no element from $B_1$ is chosen in $U_k$
    is at most $(1 - p_k)^{\epsilon |I|/2}$. Substituting for $p_k$, we can conclude that, with
    probability $1 - o(1)$, at least one element $i'$ is chosen from the set $B_1$.

    Now, $\restrict{\mu}{I}(i) \ge (1 + \epsilon/3)\restrict{\mu}{I}(i')$.  Hence,
    $\restrict{\mu}{I \cap U_k}(i) \ge (1 + \epsilon/3) \restrict{\mu}{I \cap U_k}(i')$. This
    implies that $\lVert \restrict{\mu}{I \cap U_k} - \mathcal{U}_{I \cap U_k} \rVert_\infty \ge
    \frac{\epsilon}{20|I \cap U_k|}$. The algorithm will reject with high probability in Step
    \ref{step:final-answer}, unless it has already rejected in Step \ref{step:wt-too-few}.
  
  \item Now, we consider the case where $j > 2$ is such that $\left| B_j \right| \ge
    \frac{\epsilon^2 |I|}{96\left( 1 + \epsilon/3 \right)^j \log |I|}$, and $(1 + \epsilon/3)^j >
    \log^6 n$. Let $k = \max\left\{ 0, \lfloor \log\left( \frac{|I|}{4(1 + \epsilon/3)^j \log^2 n}
  \right) \rfloor \right\}$. Then, for this value of $k$, $p_k \ge \min \left\{ 1, \frac{2(1 +
    \epsilon/3)^j \log^2 n}{|I|} \right\}$. Also, for this value of $k$, $p_k \le \min\left\{ 1,
      \tfrac{4(1 + \epsilon/3)^j\log^2 n}{|I|}\right\}$. With probability at least $1 -
      \exp(-O(\log^5 n))$, $|U_k \cap I| \ge \log^8 n$, for this value of $k$.

    Furthermore, the probability that $B_j \cap U_k$ is empty is $(1 - p_k)^{|B_j|}$. Substituting
    the values of $|B_j|$ and $p_k$, we get that $\Pr[B_j \cap U_k = \emptyset] \le \exp(-\epsilon^2
    \log n/48)$. Therefore, with probability at least $1 - \exp(-\epsilon^2 \log n/48)$, $U_k$
    contains an element of $B_j$.

    Let $i \in B_j \cap U_k$. Since $i \in  B_j$, from Lemma \ref{lem:bucketing} we know that
    $\restrict{\mu}{I}(i) = \frac{\mu(i)}{\mu(I)} \ge \frac{(1 + \epsilon)^{j-1}}{|I|}$.  From this
    bound, we get that $\restrict{\mu}{U_k}(i) \ge \frac{(1 + \epsilon)^{j-1} \mu(I)}{|I|
    \mu(U_k)}$. The expected value of $\mu(U_k)$ is $p_k$. By Markov's inequality, with probability
    at least $9/10$, $\mu(U_k) \le 10p_k$. Therefore, $\restrict{\mu}{U_k}(i) \ge
    \frac{(1+\epsilon)^{j - 1} \gamma}{10|I|p_k} \ge \frac{\gamma}{40(1 + \epsilon/3)\log^2 n}$. The
    probability that $i$ is sampled at most once in this case is at most $\log^3 (n)\left( 1 -
    \frac{\gamma}{40(1 + \epsilon/3)\log^2 n} \right)^{\log^3 (n)-1}$. Therefore, with probability
    at least $1 - o(1)$, $i$ is sampled at least twice and the tester rejects at Step
    \ref{step:wt-collision}.
\end{enumerate}
\end{proof}

\begin{proof}[Proof of Lemma \ref{lem:nonadaptive-uniformity-interval}]
 Given the input values $(\mu,I,m,\gamma,\epsilon,\delta)$, we iterate Algorithm \ref{alg:weakly-tolerant} $O(1/\delta)$ independent times with input values $(\mu,I,\gamma,\epsilon)$ (we may ignore $m$ here), and take the majority vote. The sample complexity is evident from the description of the algorithm. If indeed $\mu(I)\ge\gamma$ then Lemma \ref{lem:wt-complete} and Lemma \ref{lem:wt-sound} provide that every round gives the correct answer with probability at least $2/3$, making the majority vote correct with probability at least $1-\delta$.
\end{proof}


\section{Introducing properties characterized by atlases}\label{sec:atlas-int}

In this section, we give a testing algorithm for properties characterized by atlases, which we
formally define next. We will show in
the next subsection that distributions that are $L$-decomposable are, in particular,
characterized by atlases. First we start with the definition of an inventory.

\begin{definition}[inventory]
	Given an interval $I=[a,b]\subseteq [n]$ and a real-valued function $\nu:[a,b]\to [0,1]$, the {\em
		inventory} of $\nu$ over $[a,b]$ is the multiset $M$ corresponding to $(\nu(a),\ldots,\nu(b))$. 
\end{definition}

That is, we keep count of the function values over the interval including repetitions, but ignore
their order. In particular, for a distribution $\mu=(p_1,\dots,p_n)$ over $[n]$, the inventory
of $\mu$ over $[a,b]$ is the multiset $M$ corresponding to $(p_a,\dots,p_b)$.


\begin{definition}[atlas]
	Given a distribution $\mu$ over $[n]$, and an interval partition $\mathcal I=(I_1,\ldots,I_k)$ of
	$[n]$, the {\em atlas} $\mathcal{A}$ of $\mu$ over $\mathcal I$ is the ordered pair
	$(\mathcal I,\mathcal M)$, where $\mathcal M$ is the sequence of multisets
	$(M_1,\ldots,M_k)$ so that $M_j$ is the inventory of $\mu$ over $I_j$ for every $j\in [k]$. In this
	setting, we also say that $\mu$ {\em conforms} to $\mathcal A$.
\end{definition}

We note that there can be many distributions over $[n]$ whose atlas is the same.  We will
also denote by an atlas $\mathcal A$ any ordered pair $(\mathcal I,\mathcal M)$ where $\mathcal I$
is an interval partition of $[n]$ and $\mathcal M$ is a sequence of multisets of the
same length, so that the total sum of all members of all multisets is $1$. It is a simple observation
that for every such $\mathcal A$ there exists at least one distribution that conforms to it. The {\em length} of an
atlas $|\mathcal A|$ is defined as the shared length of its interval partition and sequence of multisets.

We now define what it means for a property to be characterized by atlases, and state our main theorem concerning such properties.

\begin{definition}
	For a function $k:\mathbb R^+ \times \mathbb N \to\mathbb N$, we say that a property of
	distributions $\mathcal{C}$ is {\em $k$-characterized by atlases} if for every $n\in\mathbb N$ and
	every $\epsilon>0$ we have a set $\mathbb A$ of atlases of lengths bounded by $k(\epsilon,n)$, so
	that every distribution $\mu$ over $[n]$ satisfying $\mathcal{C}$ conforms to some $\mathcal A\in
	\mathbb A$, while on the other hand no distribution $\mu$ that conforms to any $\mathcal A\in\mathbb
	A$ is $\epsilon$-far from satisfying $\mathcal{C}$.
\end{definition}

\begin{theorem}\label{thm:atlas-test}
	If $\mathcal{C}$ is a property of distributions that is $k$-characterized by atlases, then
	for any $\epsilon>0$ there is an adaptive conditional testing algorithm for $\mathcal{C}$
	with query complexity $k(\epsilon/5, n) \cdot \poly(\log n,1/\epsilon)$ (and error probability bound $1/3$).
\end{theorem}

\subsection{Applications and examples}
We first show that $L$-decomposable properties are in particular characterized by atlases.
\begin{lemma}
	If $\mathcal{C}$ is a property of distributions that is $L$-decomposable, then $\mathcal{C}$ is
	$k$-characterized by atlases, where $k(\epsilon, n) = L(\epsilon/3,n)$.
	\label{lem:shape-to-atlas}
\end{lemma}
\begin{proof}
	Every distribution $\mu\in \mathcal{C}$ that is supported over $[n]$ defines an atlas in
	conjunction with the interval partition of the $L$-decomposition of $\mu$ for $L = L(\gamma, n)$.
	Let $\mathbb{A}$ be the set of all such atlases. We will show that $\mathcal{C}$ is
	$L(3\gamma,n)$-characterized by $\mathbb{A}$.
	
	Let $\mu\in \mathcal{C}$. Since $\mu$ is $L$-decomposable, $\mu$ conforms to the
	atlas given by the $L$-decomposition and it is in $\mathbb{A}$ as defined above.
	
	Now suppose that $\mu$ conforms to an atlas $\mathcal{A}\in \mathbb{A}$, where
	$\mathcal{I} = (I_1,\dots,I_\ell)$ is the sequence of intervals.  By the construction of $\mathbb
	A$, there exists a distribution $\chi\in \mathcal{C}$ that conforms with $\mathcal{A}$. Now, for
	each $j \in [\ell]$ such that $\mu(I_j) \le \gamma/L$, we have (noting that $\chi(I_j) =
	\mu(I_j)$)
	\begin{equation}
	\sum_{i\in I_j} |\mu(i) - \chi(i)|  \le \sum_{i\in I_j}\mu(i) + \sum_{i\in I_j}\chi(i) 
	\le 2\mu(I_j) \le \frac{2\gamma}{\ell}.
	\label{eqn:light-atlas}
	\end{equation}
	Noting that $\mu$ and $\chi$ have the same maximum and minimum over $I_j$ (as they have the same
	inventory), for each $j\in [\ell]$ and $i\in I_j$, we know that $|\mu(i) - \chi(i)| \le \max_{i\in
		I_j}\mu(i) - \min_{i\in I_j} \mu(i)$. Therefore, for all $j\in [\ell]$ such that $\max_{i\in
		I_j}\mu(i) \le (1+\gamma) \min_{i\in I_j} \mu(i)$, $|\mu(i) - \chi(i)| \le \gamma \min_{i\in
		I_j} \mu(i)$. Therefore,
	\begin{align}
	\sum_{i\in I_j} |\mu(i) - \chi(i)| &\le |I_j| \gamma \min_{i\in I_j} \mu_j(i) \le \gamma
	\mu_j(I_j).
	\label{eqn:heavy-atlas}
	\end{align}
	Finally, recall that since $\mathcal{A}$ came from an $L$-decomposition of $\chi$, all intervals
	are covered by the above cases. Summing up Equations \ref{eqn:light-atlas} and
	\ref{eqn:heavy-atlas} for all $j \in \left\{ 1, 2, \ldots, \ell\right\}$, we obtain $d(\mu,\chi)
	\le 3\gamma$.
\end{proof}


Note that atlases characterize also properties that do not have shape restriction. The following is
a simple observation.  

\begin{observation}
	If $\mathcal{C}$ is a property of distributions that is symmetric over $[n]$, then $\mathcal{C}$
	is $\mathbf{1}$-characterized by atlases.
\end{observation}

It was shown in Chakraborty et al \cite{ChakrabortyFGM13} that such properties are efficiently
testable with conditional queries, so Theorem \ref{thm:atlas-test} in particular generalizes this result.
Also, the notion of characterization by atlases provides a natural model for tolerant testing, as we
will see in the next subsection.

%


\section{Atlas characterizations and tolerant Testing}\label{sec:atlas-tol}

We now show that for all properties of distributions that are characterized by atlases, there are
efficient \emph{tolerant testers} as well. In~\cite{CanonneDGR15}, it was shown that for a large
property of distribution properties that have ``semi-agnostic'' learners, there are efficient tolerant
testers.  In this subsection, we show that when the algorithm is given conditional query access,
there are efficient tolerant testers for the larger class of properties that are characterized by
atlases, including decomposable properties that otherwise do not lend themselves to tolerant testing.

The mechanism presented here will also be used in the proof of Theorem \ref{thm:atlas-test} itself.
First, we give a definition of tolerant testing. We note that the definition extends naturally to algorithms
that make conditional queries to a distribution.

\begin{definition}
	Let $\mathcal{C}$ be any property of probability distributions. An {\em $(\eta,\epsilon)$-tolerant
		tester} for $\mathcal{C}$ with query complexity $q$ and error probability $\delta$,
	is an algorithm that samples $q$ elements
	$x_1,\dots,x_q$ from a distribution $\mu$, accepts with probability at least $1 - \delta$ if
	$d(\mu, \mathcal{C}) \le \eta$, and rejects with probability at least $1 - \delta$ if $d(\mu,
	\mathcal{C}) \ge \eta + \epsilon$.
\end{definition}

In~\cite{CanonneDGR15}, they show that for every $\alpha>0$, there is an $\epsilon>0$ that depends on
$\alpha$, such that there is an $(\epsilon,\alpha - \epsilon)$-tolerant tester for certain
shape-restricted properties. On the other hand, tolerant testing using unconditional queries for
other properties, such as the ($\mathbf 1$-decomposable) property of being uniform, require $\Omega
\left( n/\log n\right)$ many samples (\cite{ValiantV10}).  We prove that, in the presence of
conditional query access, there is an $(\eta,\epsilon)$-tolerant tester for every $\eta,\epsilon>0$
such that $\eta+\epsilon<1$ for all properties of probability distributions that are characterized
by atlases.

We first present a definition and prove an easy lemma that will be useful later on.

\begin{definition}
	Given a partition $\mathcal I=(I_1,\ldots,I_k)$ of $[n]$, we say that a permutation $\sigma:[n]\to
	[n]$ is {\em $\mathcal I$-preserving} if for every $1\leq j\leq k$ we have $\sigma(I_j)=I_j$.
\end{definition}

\begin{lemma}
	Let $\chi$ and $\chi'$ be two distributions, supported on $[n]$, both of which conform to an atlas
	$\mathcal{A} = (\mathcal{I}, \mathcal{M})$. If $\mathcal{A}' = (\mathcal{I}, \mathcal{M}')$ is
	another atlas with the same interval partition as $\mathcal{A}$, such that $\chi$ is
	$\epsilon$-close to conforming to $\mathcal{A}'$, then $\chi'$ is also $\epsilon$-close to
	conforming to $\mathcal{A}'$.
	\label{lem:same-distance}
\end{lemma}
\begin{proof}
	It is an easy observation that there exists an $\mathcal{I}$-preserving permutation
	$\sigma$ that moves $\chi$ to $\chi'$. Now let $\mu$ be the distribution that conforms to $\mathcal{A}'$
	such that $d(\mu,\chi) \le \epsilon$, and let $\mu'$ be the distribution that results from
	having $\sigma$ operate on $\mu$. It is not hard to see that $\mu'$ conforms to $\mathcal{A}'$
	(which has the same interval partition as $\mathcal{A}$), and that it is $\epsilon$-close to $\chi'$.
\end{proof} 

For a property $\mathcal{C}$ of distributions that is $k$-characterized by atlases,
let $\mathcal{C}_\eta$ be the property of distributions
of being $\eta$-close to $\mathcal{C}$.  The following lemma states that $\mathcal{C}_\eta$  is also
k-characterized by atlases. This lemma will also be important for us outside the context of tolerant
testing per-se. 

\begin{lemma}
	Let $\mathcal{C}$ be a property of probability distributions that is $k(\epsilon,n)$-characterized by
	atlases. For any $\eta>0$, let $\mathcal{C}_{\eta}$ be the set of all probability distributions
	$\mu$ such that $d(\mu, \mathcal{C}) \le \eta$.  Then, there is a set $\mathbb{A}_\eta$ of atlases, of
	length at most $k$, such that every $\mu \in \mathcal{C}_\eta$ conforms to at least one atlas in
	$\mathbb{A}_\eta$, and every distribution that conforms to an atlas in $\mathbb{A}_\eta$ is $\eta +
	\epsilon$-close to $\mathcal{C}$, and is $\epsilon$-close to $\mathcal{C}_\eta$.
	\label{lem:close-characterized}
\end{lemma}
\begin{proof}
	Since $\mathcal{C}$ is $k$-characterized by atlases, there is a set of atlases $\mathbb{A}$ of
	length at most $k(\epsilon,n)$ such that for each $\mu\in \mathcal{C}$, there is an atlas
	$\mathcal{A}\in \mathbb{A}$ to which it conforms, and any $\chi$ that conforms to an atlas
	$\mathcal{A}\in \mathbb{A}$ is $\epsilon$-close to $\mathcal{C}$. Now, let
	$\mathbb{A}_\eta$ be obtained by taking each atlas
	$\mathcal{A}\in \mathbb{A}$, and adding all atlases, with the same interval partition,
	corresponding to distributions that are $\eta$-close to conforming to $\mathcal{A}$. 
	
	First, note that the new atlases that are added have the same interval partitions as atlases in
	$\mathbb{A}$, and hence have the same length bound $k(\epsilon, n)$. To complete the proof of the
	lemma, we need to prove that every $\mu\in \mathcal{C}_\eta$ conforms to some atlas in
	$\mathbb{A}_\eta$, and that no distribution that conforms to $\mathbb{A}_\eta$ is $\eta + \epsilon$-far
	from $\mathcal{C}$.
	
	Take any $\mu\in \mathcal{C}_\eta$. There exists some distribution $\mu' \in \mathcal{C}$ such
	that $d(\mu, \mu') \le \eta$. Since $\mathcal{C}$ is $k$-characterized by atlases, there is some
	atlas $\mathcal{A} \in \mathbb{A}$ such that $\mu'$ conforms with $\mathcal{A}$. Also, observe
	that $\mu$ is $\eta$-close to conforming to $\mathcal{A}$ through $\mu'$. Therefore, there is an
	atlas $\mathcal{A}'$ with the same interval partition as $\mathcal{A}$ that was added in
	$\mathbb{A}_\eta$, which is the atlas corresponding to the distribution $\mu$. Hence, there is an
	atlas in $\mathbb{A}_\eta$ to which $\mu$ conforms.
	
	Conversely, let $\chi$ be a distribution that conforms with an atlas $\mathcal{A}' \in
	\mathbb{A}_\eta$. From the construction of $\mathbb{A}_\eta$, we know that there is an atlas $\mathcal{A}
	\in \mathbb{A}$ with the same interval partition as $\mathcal{A}'$, and there is a
	distribution $\chi'$ that conforms to $\mathcal{A}'$ and is $\eta$-close to conforming to
	$\mathcal{A}$. Therefore, by Lemma \ref{lem:same-distance} $\chi$ is also $\eta$-close to
	conforming to $\mathcal{A}$. Let $\mu'$ be the distribution conforming to $\mathcal{A}$ such that
	$d(\chi, \mu') \le \eta$. Since $\mu'$ conforms to an atlas $\mathcal{A} \in \mathbb{A}$, $d(\mu',
	\mathcal{C}) \le \epsilon$.  Therefore, by the triangle inequality, $d(\chi, \mathcal{C}) \le \eta
	+ \epsilon$.
	
	This also implies that $d(\chi,\mathcal{C}_\eta) \le \epsilon$ by considering $\tilde{\chi} =
	(\eta \chi+ \epsilon \tilde{\mu})/(\epsilon + \eta)$ where $\tilde{\mu}$ is the distribution in
	$\mathcal{C}$ that is $\epsilon + \eta$-close to $\chi$.  Note that $\tilde{\chi}$ is $\eta$-close
	to $\tilde{\mu}$ and $\epsilon$-close to $\chi$.
\end{proof}

Using Lemma~\ref{lem:close-characterized} we get the following corollary of Theorem \ref{thm:atlas-test}
about tolerant testing of distributions characterized by atlases.

\begin{corollary}
	Let $\mathcal{C}$ be a property of distributions that is $k$-characterized by atlases. For every
	$\eta,\epsilon >0$ such that $\eta+\epsilon<1$, there is an $(\eta,\epsilon)$-tolerant tester for
	$\mathcal{C}$ that takes $k(\epsilon/5,n)\cdot\poly(\log n,1/\epsilon)$ conditional
	samples and succeeds with probability at least $2/3$.
\end{corollary}

%
%
%
%


\section{Some useful lemmas about atlases and characterizations}\label{sec:atlas-lem}

We start with a definition and a lemma, providing an alternative equivalent definition of properties $k$-characterizable by atlases

\begin{definition}[permutation-resistant distributions]
	For a function $k:\mathbb R^+ \times \mathbb N \to\mathbb N$, a property $\mathcal{C}$ of
	probability distributions is {\em $k$-piecewise permutation resistant} if for every $n\in\mathbb N$,
	every $\epsilon>0$, and every distribution $\mu$ over $[n]$ in $\mathcal{C}$, there exists a
	partition $\mathcal I$ of $[n]$ into up to $k(\epsilon, n)$ intervals, so that every $\mathcal
	I$-preserving permutation of $[n]$ transforms $\mu$ into a distribution that is $\epsilon$-close to
	a distribution in $\mathcal{C}$.
\end{definition}

\begin{lemma}
	For $k: \mathbb R^+ \times \mathbb N \to \mathbb N$, a property $\mathcal{C}$ of probability
	distributions over $[n]$ is $k$-piecewise permutation resistant if and only if it is
	$k$-characterized by atlases.
	\label{lem:piecewise-atlas}
\end{lemma}
\begin{proof}
	If $\mathcal{C}$ is $k$-piecewise permutation resistant, then for each distribution $\mu\in
	\mathcal{C}$, there exists an interval partition $\mathcal{I}$ of $[n]$ such that every
	$\mathcal{I}$-preserving permutation of $[n]$ transforms $\mu$ into a distribution that is
	$\epsilon$-close to $\mathcal{C}$. Each distribution $\mu$ thus gives an atlas over $\mathcal{I}$,
	and the collection of these atlases for all $\mu\in \mathcal{C}$ characterizes the property
	$\mathcal{C}$.  Therefore, $\mathcal{C}$ is $k$-characterized by atlases.
	
	Conversely, let $\mathcal{C}$ be a property of distributions that are $k$-characterized by atlases
	and let $\mathbb{A}$ be the set of atlases. For each $\mu\in C$, let $\mathcal{A}_\mu$ be the
	atlas in $\mathbb{A}$ that characterizes $\mu$ and let $\mathcal{I}_\mu$ be the interval partition
	corresponding to this atlas. Now, every $\mathcal{I}_\mu$-preserving permutation $\sigma$ of $\mu$
	gives a distribution $\mu_\sigma$ that has the same atlas $\mathcal{A}_\mu$. Since $\mathcal{C}$
	is $k$-characterized by atlases, $\mu_\sigma$ is $\epsilon$-close to $\mathcal{C}$.  Therefore,
	$\mathcal{C}$ is $k$-piecewise permutation resistant as well.
\end{proof}

We now prove the following lemma about $\epsilon/k$-fine partitions of distributions characterized by atlases, having a similar flavor as Lemma \ref{lem:small-wt-non-uniform} for $L$-decomposable properties. Since we cannot avert a $\poly(\log n)$ dependency anyway, for simplicity we use the $1$-parameter variant of fine partitions.

\begin{lemma}
	Let $\mathcal{C}$ be a property of distributions that is $k(\epsilon, n)$-characterized by atlases
	through $\mathbb A$. For any $\mu \in \mathcal{C}$, any $\epsilon/k$-fine interval partition
	$\mathcal{I}'$ of $\mu$, and the corresponding atlas $\mathcal{A}' = (\mathcal{I}', \mathcal{M}')$
	for $\mu$ (not necessarily in $\mathbb A$), any distribution $\mu'$ that conforms to
	$\mathcal{A}'$ is $3\epsilon$-close to $\mathcal{C}$.
	\label{lem:atlas-pulled}
\end{lemma}
\begin{proof}
	Let $\mathcal{A} = (\mathcal{I}, \mathcal{M})$ be the atlas from $\mathbb A$ to which $\mu$
	conforms, and let $\mathcal{I}' = (I'_1, I'_2, \ldots, I'_r)$ be an $\epsilon/k$-fine interval
	partition of $\mu$.  Let $\mathcal{P} \subseteq \mathcal{I}'$ be the set of intervals that
	intersect more than one interval in $\mathcal{I}$. Since $\mathcal{I}'$ is $\epsilon/k$-fine, and
	the length of $\mathcal{A}$ is at most $k$, $\mu(\bigcup_{I'_j \in \mathcal{P}} I'_j) \le
	\epsilon$ (note that $\mathcal{P}$ cannot contain singletons). Also, since $\mu'$ conforms to
	$\mathcal{A}'$, we have $\mu'(\bigcup_{I'_j \in \mathcal{P}} I'_j) \le \epsilon$.
	
	Let $\tilde{\mu}$ be a distribution supported over $[n]$ obtained as follows: For each interval
	$I'_j \in \mathcal{P}$, $\tilde{\mu}(i) = \mu(i)$ for every $i \in I'_j$. For each interval $I'_j
	\in \mathcal{I}' \setminus \mathcal{P}$, $\tilde{\mu}(i) = \mu'(i)$ for every $i \in I'_j$. Note that
	that the inventories of $\tilde{\mu}$ and $\mu$ are identical over any $I'_j$ in $\mathcal{I}' \setminus
	\mathcal{P}$. From this it follows that $\tilde{\mu}$ also conforms to $\mathcal{A}$, and in
	particular $\tilde{\mu}$ is a distribution. To see this, for any $I_j$ in $\mathcal{I}$ partition
	it to its intersection with the members of $\mathcal{I}' \setminus \mathcal{P}$ contained in it,
	and all the rest. For the former we use that $\mu$ and  $\mu'$ have the same inventories, and for
	the latter we specified that $\tilde{\mu}$ has the same values as $\mu$.
	
	Since $\mu'$ and $\tilde{\mu}$ are identical at all points except those in $\mathcal{P}$, we have
	$d(\mu', \tilde{\mu}) \le 2\epsilon$.  Furthermore, $d(\tilde{\mu}, \mathcal{C}) \le \epsilon$
	since $\tilde{\mu}$ conforms to $\mathcal{A} \in \mathbb A$. Therefore, by the triangle
	inequality, $d(\mu', \mathcal{C}) \le 3\epsilon$. 
\end{proof}

The main idea of our test for a property of distributions $k$-characterized by atlases, starts with a
$\gamma/k$-fine partition $\mathcal{I}$ obtained by through Algorithm \ref{alg:pulled-partition}.
We then show how to compute
an atlas $\mathcal{A}$ with this interval partition such that there is a distribution
$\mu_\mathcal{I}$ that conforms to $\mathcal{A}$ that is close to $\mu$. We use the
$\epsilon$-trimming sampler from \cite{ChakrabortyFGM13} to obtain such an atlas corresponding to
$\mathcal{I}$.  To test if $\mu$ is in $\mathcal{C}$, we show that it is sufficient to check if
there is some distribution conforming to $\mathcal{A}$ that is close to a distribution in $\mathcal{C}$.


\section{An adaptive test for properties characterized by atlases}\label{sec:atlas-test}

Our main technical lemma, which we state here and prove in Section \ref{sec:atlas-learn},
is the following possibility of ``learning'' an atlas
of an unknown distribution for an interval partition $\mathcal{I}$, under the conditional sampling
model.

\begin{lemma}\label{lem:atlas-learn}
	Given a distribution $\mu$ supported over $[n]$, and a partition $\mathcal{I} = (I_1,
	I_2, \ldots, I_r)$, using $r \cdot \poly(\log n, 1/\epsilon, \log(1/\delta))$ conditional
	samples from $\mu$ we can construct, with probability at least $1 - \delta$, an atlas for some
	distribution $\mu_\mathcal{I}$ that is $\epsilon$-close to $\mu$.
\end{lemma}

First, we show how this implies Theorem \ref{thm:atlas-test}. To prove it, we give as Algorithm \ref{alg:atlas-test} a formal description of the test.

\begin{algorithm}
	[htpb]
	\caption{Adaptive conditional tester for properties $k$-characterized by atlases}
	\label{alg:atlas-test}
	
	\KwInput{A distribution $\mu$ supported over $[n]$, a function $k:(0,1] \times \mathbb N \to
		\mathbb N$, accuracy parameter $\epsilon > 0$, a property $\mathcal{C}$ of distributions that is
		$k$-characterized by the set of atlases $\mathbb A$}
	
	Use Algorithm \ref{alg:pulled-partition} with input values $(\mu,\epsilon/5k(\epsilon/5,n),1/6)$ to obtain a partition $\mathcal{I}$ with $|\mathcal{I}|\le 20k(\epsilon/5,n)\log(n)\log(1/\epsilon)/\epsilon$
	\label{step:atlas-fine}
	
	Use Lemma \ref{lem:atlas-learn} with accuracy parameter $\epsilon/5$ and error parameter
	$1/6$ to obtain an atlas $\mathcal{A}_{\mathcal{I}'}$ corresponding to $\mathcal{I}'$
	\label{step:learn-atlas}
	
	\lIf* {there exists $\chi \in \mathcal{C}$ that is $\epsilon/5$-close to conforming to
		$\mathcal{A}_{\mathcal{I}'}$} {$\accept$} \lElse {$\reject$}
	\label{step:atlas-decide}
\end{algorithm}

\begin{lemma}
	[completeness]
	Let $\mathcal{C}$ be a property of distributions that is $k$-characterized by atlases, and let $\mu$
	be any distribution supported over $[n]$. If $\mu \in
	\mathcal{C}$, then with probability at least $2/3$ Algorithm \ref{alg:atlas-test} accepts.
	\label{lem:atlas-complete}
\end{lemma}
\begin{proof}
	In Step \ref{step:learn-atlas}, with probability at least $5/6>2/3$,
	we get an atlas $\mathcal{A}_{\mathcal{I}'}$ such that there is a distribution
	$\mu_{\mathcal{I}'}$ that is $\epsilon/5$-close to $\mu$. Step \ref{step:atlas-decide} then
	accepts on account of $\chi = \mu$.
\end{proof}

\begin{lemma}
	[soundness]
	Let $\mathcal{C}$ be a property of distributions that is $k$-characterized by atlases, and let $\mu$
	be any distribution supported over $[n]$. If $d(\mu,\mathcal{C}) > \epsilon$,
	then with probability at least $2/3$ Algorithm \ref{alg:atlas-test} rejects.
	\label{lem:atlas-sound}
\end{lemma}
\begin{proof}
	With probability at least $2/3$, for $k=k(\epsilon/5,n)$ we get an $\epsilon/5k$-fine partition
	$\mathcal{I}'$ in Step \ref{step:atlas-fine}, as well as an atlas $\mathcal{A}_{\mathcal{I}'}$
	in Step \ref{step:learn-atlas} such that there is a distribution $\mu_{\mathcal{I}'}$ conforming
	to it that is $\epsilon/5$-close to $\mu$.  
	
	Suppose that the algorithm accepted in Step \ref{step:atlas-decide} on account of $\chi \in
	\mathcal{C}$.  Then there is a $\chi'$ that is $\epsilon/5$-close to $\chi$ and conforms to
	$\mathcal{A}_{\mathcal{I}'}$.  By Lemma \ref{lem:close-characterized}, the property of being
	$\epsilon/5$-close to $\mathcal{C}$ is itself $k$-characterized by atlases. Let $\mathbb A_{\epsilon/5}$ be the
	collection of atlases characterizing it.  Using Lemma \ref{lem:atlas-pulled} with $\chi'$ and
	$\mathbb A_{\epsilon/5}$, we know that $\chi'$ is $3\epsilon/5$-close to some $\chi$, which is
	in $\mathcal{C}_{\epsilon/5}$ and thus
	$\epsilon/5$-close to $\mathcal{C}$.  Since $\chi'$ is also $\epsilon/5$-close to $\mu$, we
	obtain that $\mu$ is $\epsilon$-close to $\mathcal{C}$ by the triangle inequality, contradicting
	$d(\mu, \mathcal{C}) > \epsilon$.
\end{proof}

\begin{proof}[Proof of Theorem \ref{thm:atlas-test}]
	Given a distribution $\mu$, supported on $[n]$, and a property $\mathcal{C}$ of distributions that is $k$-characterized by
	atlases, we use Algorithm \ref{alg:atlas-test}. The correctness follows from Lemmas
	\ref{lem:atlas-complete} and \ref{lem:atlas-sound}. The number of samples made in Step \ref{step:atlas-fine}
	is clearly dominated by the number of samples in Step \ref{step:learn-atlas}, which is
	$k(\epsilon/5,n) \cdot \poly(\log n, 1/\epsilon)$. 
\end{proof}


\section{Constructing an atlas for a distribution}\label{sec:atlas-learn}

Before we prove Lemma~\ref{lem:atlas-learn}, we will define the notion of value-distances and
prove lemmas that will be useful for the proof of the theorem.

\begin{definition}[value-distance]
	Given two multisets $A$,$B$ of real numbers, both of the same size (e.g. two inventories over an
	interval $[a,b]$), the {\em value-distance} between them is the minimum $\ell_1$ distance between
	a vector that conforms to $A$ and a vector that conforms to $B$.
\end{definition}

The following observation gives a simple method to calculate the value-distances between two
multisets $A$ and $B$.

\begin{observation}
	The value-distance between $A$ and $B$ is equal to the $\ell_1$ distance between the two vectors
	conforming to the respective sorting of the two multisets.
\end{observation}
\begin{proof}
	Given two vectors $v$ and $w$ corresponding to $A$ and $B$ achieving the value-distance, first
	assume that the smallest value of $A$ is not larger than the smallest value of $B$. Assume without
	loss of generality (by permuting both $v$ and $w$) that $v_1$ is of the smallest value among those
	of $A$. It is not hard to see that, in this case, one can make $w_1$ to be of the smallest value among
	those of $B$ without increasing the distance (by swapping the two values), and from here one can
	proceed by induction over $|A|=|B|$.
\end{proof}

We now prove two lemmas that will be useful for the proof of Lemma~\ref{lem:atlas-learn}.
\begin{lemma}\label{lem:valuedist}
	Let $A$ and $B$ be two multisets of the same size, both with members whose values range in
	$\{0,\alpha_1,\ldots,\alpha_r\}$. Let $m_j$ be the number of appearances of $\alpha_j$ in $A$, and
	$n_j$ the corresponding number in $B$. If $m_j\leq n_j$ for every $1\leq j\leq r$, then the
	value-distance between $A$ and $B$ is bounded by $\sum_{j=1}^r(n_j-m_j)\alpha_j$.
\end{lemma}

\begin{proof}
	Let $v_A=\{a_1,\dots,a_l\}$ and $v_B = \{b_1,\dots,b_l\}$ be two vectors such that $a_1=\dots
	=a_{m_1}=\alpha_1$, $a_{m_1+1}=\dots=a_{n_1}=0$ and $b_1=\dots =b_{n_1}=\alpha_1$, and similarly
	for $j\in\{1,\dots,r-1\}$, $a_{n_j+1}=\dots=a_{n_j+m_{j+1}}=\alpha_{j+1}$,
	$a_{n_j+m_{j+1}+1}=\dots=a_{n_j+n_{j+1}}=0$ and $b_{n_j+1}=\dots=b_{n_j+n_{j+1}}=\alpha_{j+1}$.
	For $k > \sum_{j=1}^r n_j$, we set $a_k = b_k = 0$.  The vectors $v_A$ and $v_B$ conform to the
	multisets $A$ and $B$ respectively, and the $\ell_1$ distance between the two vectors is
	$\sum_{j=1}^r(n_j-m_j) \alpha_j$, so the lemma follows.
\end{proof}

\begin{lemma}
	Let $\mu$ be a probability distribution over $\{1,\dots,n\}$, and let $\tilde{\mu}$ be a vector of
	size $n$ where each entry is a real number in the interval $[0,1]$, such that
	$\sum_{i\in[n]}|\mu(i) - \tilde{\mu}(i)| \le \epsilon$. Let $\hat{\mu}$ be a probability
	distribution over $\{1,\dots,n\}$ defined as $\hat{\mu}(i) = \tilde{\mu}(i)/ \sum_{i\in
		[n]}\tilde{\mu}(i)$ for all $i\in [n]$. Then $\sum_{i\in [n]}|\mu(i) - \hat{\mu}(i)| \le
	5\epsilon$.
	\label{lem:normalize-close}
\end{lemma}
\begin{proof}
	We have $|\sum_{i\in [n]}\mu(i) - \sum_{i\in [n]}\tilde{\mu}(i)| \le \epsilon$. Therefore, $1 -
	\epsilon \le \sum_{i\in [n]}\tilde{\mu}(i) \le 1 + \epsilon$. If $\sum_{i\in [n]} \tilde{\mu}(i) <
	1$, then $\hat{\mu}(i) \le \tilde{\mu}(i)/ (1 - \epsilon)$ and $\hat{\mu}(i) > \tilde{\mu}(i)$.
	Therefore, $\hat{\mu}(i) \le (1 + 2\epsilon)\tilde{\mu}(i)$ and hence $0 \le \hat{\mu}(i) -
	\tilde{\mu}(i) \le 2\epsilon \tilde{\mu}(i)$. If $\sum_{i\in [n]} \tilde{\mu}(i) \ge 1$, then
	$\hat{\mu}(i) \ge \tilde{\mu}(i)/(1 + \epsilon) \ge (1 - \epsilon)\tilde{\mu}(i)$ and
	$\hat{\mu}(i) \le \tilde{\mu}(i)$.  Therefore $0 \le \tilde{\mu}(i) - \hat{\mu}(i) \le \epsilon
	\tilde{\mu}(i)$.  Therefore $|\tilde{\mu}(i) - \hat{\mu}(i)| \le 2\epsilon \tilde{\mu}(i)$, in all
	cases, for all $i$.
	
	Now, $\sum_{i\in [n]} |\mu(i) - \hat{\mu}(i)| \le \sum_{i\in [n]} |\mu(i) - \tilde{\mu}(i)| +
	\sum_{i\in [n]} |\tilde{\mu}(i) - \hat{\mu}(i)|$.  Since $\sum_{i\in [n]} |\tilde{\mu}(i) -
	\hat{\mu}(i)| \le 2\epsilon\sum_{i\in [n]}\tilde{\mu}(i) \le 2\epsilon (1 + \epsilon)$, we get
	that $\sum_{i\in [n]} |\mu(i) - \hat{\mu}(i)| \le 5\epsilon$.
\end{proof}

Now we recall the definition of an $\epsilon$-trimming sampler from \cite{ChakrabortyFGM13}.

\begin{definition}[$\epsilon$-trimming sampler]
	An \emph{$\epsilon$-trimming sampler} with $s$ samples for a distribution $\mu$ supported over
	$[n]$, is an algorithm that has conditional query access to the distribution $\mu$ and returns $s$
	pairs of values $(r,\overline{\mu}(r))$ (for $r=0$, $\overline{\mu}(r)$ is not output) from a
	distribution $\overline{\mu}$ supported on $\{0\}\cup [n]$, such that
	$\sum_{i\in [n]}|\overline{\mu}(i)-\mu(i)|\le 4\epsilon$, and each $r$ is
	independently drawn from $\overline{\mu}$.  Furthermore, there is a set $P$ of $\poly(\log n,
	1/\epsilon)$ real numbers such that for all $i$ either $\overline{\mu}(i)= 0$ or
	$\overline{\mu}(i)\in P$. 
	\label{defn:trimming-sampler}
	%
\end{definition}

The existence of an $\epsilon$-trimming sampler with a small set of values was proved
in~\cite{ChakrabortyFGM13}. Let us formally state this lemma.

\begin{lemma}[\cite{ChakrabortyFGM13}]\label{lem:trimming}
	Given conditional query access to a distribution $\mu$ supported on $[n]$, there is an
	$\epsilon$-trimming sampler that makes $32s\cdot\epsilon^{-4} \cdot \log^5n \cdot
	\log(s\delta^{-1} \log n)$ many conditional queries to $\mu$, and returns, with probability at
	least $1-\delta$, a sequence of $s$ samples from a distribution $\overline{\mu}$, where $P =
	\{\frac{(1+\epsilon)^{i-1}\epsilon}{n} \mid 1 \le i \le k-1\}$ for $k=\log n\log(\epsilon^{-1})/
	\log^2(1 + \epsilon)$.
	%
\end{lemma}

\subsection{Proving the main lemma}

The proof of Lemma \ref{lem:atlas-learn} depends on the following technical lemma.
\begin{lemma}
	Let $\mu$ be a distribution supported over $[n]$, and let $\mathcal{P} = (P_0, P_1, P_2, \ldots,
	P_r)$ be a partition of $[n]$ into $r + 1$ subsets with the following properties.
	\begin{enumerate}
		\item For each $P_k \in \mathcal{P}$, $\mu(i) = p_k$ for every $i \in P_k$, where $p_1, \ldots,
		p_r$ (but not $p_0$) are known.
		\item Given an $i \in [n]$ sampled from $\mu$, we can find the $k$ such that $i \in P_k$.
	\end{enumerate}
	Using $s= \tfrac{6r}{\epsilon^2} \log\left( \tfrac{r}{\delta} \right)$ samples from $\mu$, with
	probability at least $1 - \delta$, we can find $m_1, \ldots, m_r$ such that $m_k \le |P_k|$ for
	all $k \in [r]$, and $\sum_{k \in [r]} p_k (|P_k| - m_k) \le 4\epsilon$.
	\label{lem:estimate}
\end{lemma}
\begin{proof}
	Take $s$ samples from $\mu$. For each $k \in [r]$, let $s_k$ be the number of samples in $P_k$,
	each with probability $p_k$.  We can easily see that $E[s_k] = sp_k|P_k|$. 
	
	If $p_k|P_k| \ge 1/r$, then by Chernoff bounds, we know that,  
	\begin{align*}
	\Pr\left[ (1 - \epsilon)E[s_k] \le s_k \le (1 + \epsilon)E[s_k] \right] \ge 1 - 2e^{-\epsilon^2
		E[s_k]/3} 
	\end{align*} 
	By the choice of $s$, with probability at least $1 - \delta/r$, $(1 - \epsilon)p_k|P_k| \le
	\tfrac{s_k}{s} \le (1+\epsilon)p_k|P_k|$.
	
	On the other hand, if
	$p_k|P_k| < 1/r$, then by Chernoff bounds we have,
	\begin{align*}
	\Pr\left[ p_k|P_k| - \tfrac{\epsilon}{r}
	\le \tfrac{s_k}{s} \le p_k|P_k| + \tfrac{\epsilon}{r} \right] &=\\
	\Pr\left[ \left(1 - \tfrac{\epsilon}{rp_k|P_k|}\right) sp_k|P_k| \le s_k \le \left(1 +
	\tfrac{\epsilon}{rp_k|P_k|} \right)sp_k|P_k| \right] &\ge 1 - 2 \exp\left( -
	\frac{\epsilon^2 s}{r^2 p_k |P_k|} \right) \\
	&\ge 1 - 2e^{-\epsilon^2 s/3r} 
	\end{align*}
	By the choice of $s$, with probability at least $1 - \delta/r$, $p_k|P_k| - \tfrac{\epsilon}{r}
	\le \tfrac{s_k}{s} \le p_k|P_k| + \tfrac{\epsilon}{r}$. 
	
	With probability at least $1 - \delta$, we get an estimate $\alpha_k = s_k/s$ for every $k \in
	[r]$ such that $\alpha_k \le \max\left\{ p_k|P_k| + \epsilon/r, p_kP_k(1 + \epsilon) \right\}$,
	and $\alpha_k \ge \min\left\{ p_k|P_k| - \epsilon/r, (1 - \epsilon)p_k|P_k| \right\}$. From now on
	we assume that $\alpha_k$ satisfies the above bounds, and define $\alpha'_k = \min\left\{ \alpha_k
	- \epsilon/r, \alpha_k/(1 + \epsilon) \right\}$. Notice that $\alpha'_k \le p_k|P_k|$.
	Furthermore, $\alpha'_k \ge \min\left\{ p_k|P_k| - 2\epsilon/r, (1 - 2\epsilon)p_k|P_k| \right\}$.
	Set $m_k = \lceil \alpha'_k/p_k\rceil$. Since $\alpha'_k/p_k \le |P_k|\in\mathbb{Z}$, we have $m_k \le |P_k|$
	for all $k$.
	
	Now, $\sum_{k \in [r]} p_k(|P_k| - m_k) \le \sum_{k \in [r]} (p_k|P| - \alpha'_k)$. Under the
	above assumption on the value of $\alpha_k$, for every $k$ such that $p_k|P_k| < 1/r$, we have
	$\alpha'_k \ge p_k|P_k| - 2\epsilon/r$. Hence, this difference is at most $2\epsilon/r$. For
	every $k$ such that $p_k|P_k| \ge 1/r$, $\alpha'_k \ge (1 - 2\epsilon)p_k|P_k|$.  For any such
	$k$, the difference $p_k|P_k| - \alpha'_k$ is at most $2\epsilon p_k|P_k|$.  Therefore, $\sum_{k
		\in [r]} p_k(|P_k| - m_k)$ is at most $4\epsilon$.
\end{proof}

\begin{proof}
	[Proof of Lemma \ref{lem:atlas-learn}]
	Given a distribution $\mu$ and an interval partition $\mathcal{I} = (I_1, I_2, \ldots, I_r)$, let
	$\overline{\mu}$ be the distribution presented by the $\epsilon/8$-trimming sampler in Lemma
	\ref{lem:trimming}. Let $I_{j,k}\subseteq [n]$ be the set of indexes $i$ such that $i \in I_j$ and
	$\overline{\mu}(i) = \tfrac{(1 + \epsilon/8)^{k-1}\epsilon}{8n}$. Thus, each interval $I_j$ in
	$\mathcal{I}$ is now split into subsets $I_{j,0}, I_{j,1}, I_{j,2}, \ldots, I_{j,\ell}$, where
	$\ell \le \log n \log(8/\epsilon)/\log^2(1 + \epsilon/8)$ and $I_{j,0}$ is the set of indexes in
	$I_j$ such that $\overline{\mu}(i) = 0$. 
	
	Using Lemma \ref{lem:trimming}, with $s = r\cdot \poly(\log n, 1/\epsilon, \log(1/\delta))$ samples from the  distribution
	$\overline{\mu}$, we can estimate, with probability at least $1 - \delta$, the values $m_{j,k}$
	such that $m_{j,k} \le |I_{j,k}|$ for all $k > 0$, and the following holds.  \[\sum_{j,k : k > 0}
	\tfrac{(1 + \epsilon/8)^{k-1}\epsilon}{8n}\left( I_{j,k} - m_{j,k} \right) \le \epsilon/10.\]
	For every $j$, let $M_{I_j}$ be the inventory provided by $m_{j,1},\ldots,m_{j,\ell}$ and
	$m_{j,0} = |I_j| - \sum_{k \in [\ell]} m_{j,k}$.  Thus, we have a inventory sequence
	$\tilde{\mathcal{M}}_\mathcal{I} = (M_{I_1}, M_{I_2}, \ldots, M_{I_r})$ that is
	$\epsilon/10$-close in value-distance to the corresponding atlas of $\overline{\mu}$ (where we
	need to add the interval $\{0\}$ to the partition to cover its entire support). Corresponding to
	$\tilde{\mathcal{M}}_{\mathcal{I}}$, there is a vector $\tilde{\mu}$ that is
	$\epsilon/10$-close to $\overline{\mu}$. Using Lemma \ref{lem:normalize-close}, we have a
	distribution $\hat{\mu}$ that is $\epsilon/2$-close to $\overline{\mu}$. Since
	the $[n]$ portion of
	$\overline{\mu}$ is $\epsilon/2$-close to $\mu$, by the triangle inequality, $\hat{\mu}$
	is $\epsilon$-close to $\mu$. Thus $\mathcal{A} = (\mathcal{I}, \mathcal{M}_\mathcal{I})$, where
	$\mathcal{M}_\mathcal{I}$ is obtained by multiplying all members of
	$\tilde{\mathcal{M}}_\mathcal{I}$ by the same factor used to produce $\hat{\mu}$ from
	$\tilde{\mu}$, is an atlas for a distribution that is $\epsilon$-close to $\mu$.
	
	We need $s\cdot \poly(\log n, \log s, 1/\epsilon, \log(1/\delta)$ conditional samples from $\mu$
	to get $s$ samples from $\overline{\mu}$. Therefore, we require $r\cdot \poly(\log n, 1/\epsilon,
	\log(1/\delta))$ conditional samples to construct the atlas $A_\mathcal{I}$.
\end{proof}


\bibliographystyle{amsalpha}

\begin{thebibliography}{DLM{\etalchar{+}}07}

\bibitem[AAK{\etalchar{+}}07]{AlonAKMR07}
Noga Alon, Alexandr Andoni, Tali Kaufman, Kevin Matulef, Ronitt Rubinfeld, and
  Ning Xie, \emph{Testing k-wise and almost k-wise independence}, Proceedings
  of the Thirty-ninth Annual ACM Symposium on Theory of Computing, STOC '07,
  ACM, 2007, pp.~496--505.

\bibitem[ACK15]{AcharyaCK15}
Jayadev Acharya, Cl{\'{e}}ment~L. Canonne, and Gautam Kamath, \emph{A chasm
  between identity and equivalence testing with conditional queries},
  Approximation, Randomization, and Combinatorial Optimization. Algorithms and
  Techniques, {APPROX/RANDOM} 2015, August 24-26, 2015, Princeton, NJ, {USA},
  2015, pp.~449--466.

\bibitem[ADK15]{AcharyaDK15}
Jayadev Acharya, Constantinos Daskalakis, and Gautam Kamath, \emph{Optimal
  testing for properties of distributions}, Advances in Neural Information
  Processing Systems 28: Annual Conference on Neural Information Processing
  Systems 2015, December 7-12, 2015, Montreal, Quebec, Canada (Corinna Cortes,
  Neil~D. Lawrence, Daniel~D. Lee, Masashi Sugiyama, and Roman Garnett, eds.),
  2015, pp.~3591--3599.

\bibitem[BFF{\etalchar{+}}01]{BatuFFKRW01}
Tugkan Batu, Lance Fortnow, Eldar Fischer, Ravi Kumar, Ronitt Rubinfeld, and
  Patrick White, \emph{Testing random variables for independence and identity},
  42nd Annual Symposium on Foundations of Computer Science, {FOCS} 2001, 14-17
  October 2001, Las Vegas, Nevada, {USA}, 2001, pp.~442--451.

\bibitem[BFR{\etalchar{+}}00]{BatuFRSW00}
Tugkan Batu, Lance Fortnow, Ronitt Rubinfeld, Warren~D. Smith, and Patrick
  White, \emph{Testing that distributions are close}, 41st Annual Symposium on
  Foundations of Computer Science, {FOCS} 2000, 12-14 November 2000, Redondo
  Beach, California, {USA}, {IEEE} Computer Society, 2000, pp.~259--269.

\bibitem[BKR04]{BatuKR04}
Tugkan Batu, Ravi Kumar, and Ronitt Rubinfeld, \emph{Sublinear algorithms for
  testing monotone and unimodal distributions}, Proceedings of the 36th Annual
  {ACM} Symposium on Theory of Computing, Chicago, IL, USA, June 13-16, 2004
  (L{\'{a}}szl{\'{o}} Babai, ed.), {ACM}, 2004, pp.~381--390.

\bibitem[Can15]{Canonne15}
Cl{\'{e}}ment~L. Canonne, \emph{A survey on distribution testing: Your data is
  big. but is it blue?}, Electronic Colloquium on Computational Complexity
  {(ECCC)} \textbf{22} (2015), 63.

\bibitem[CDGR15]{CanonneDGR15}
Cl{\'{e}}ment~L. Canonne, Ilias Diakonikolas, Themis Gouleakis, and Ronitt
  Rubinfeld, \emph{Testing shape restrictions of discrete distributions}, CoRR
  \textbf{abs/1507.03558} (2015).

\bibitem[CFGM13]{ChakrabortyFGM13}
Sourav Chakraborty, Eldar Fischer, Yonatan Goldhirsh, and Arie Matsliah,
  \emph{On the power of conditional samples in distribution testing},
  Innovations in Theoretical Computer Science, {ITCS} '13, Berkeley, CA, USA,
  January 9-12, 2013, 2013, pp.~561--580.

\bibitem[CRS15]{CanonneRS15}
Cl{\'{e}}ment~L. Canonne, Dana Ron, and Rocco~A. Servedio, \emph{Testing
  probability distributions using conditional samples}, {SIAM} J. Comput.
  \textbf{44} (2015), no.~3, 540--616.

\bibitem[Dia16]{diakonikolas2016learning}
Ilias Diakonikolas, \emph{Learning structured distributions}, Handbook of Big
  Data (2016), 267.

\bibitem[DK16]{DiakonikolasK16}
Ilias Diakonikolas and Daniel~M. Kane, \emph{A new approach for testing
  properties of discrete distributions}, CoRR \textbf{abs/1601.05557} (2016).

\bibitem[DLM{\etalchar{+}}07]{DiakonikolasLMORSW07}
Ilias Diakonikolas, Homin~K. Lee, Kevin Matulef, Krzysztof Onak, Ronitt
  Rubinfeld, Rocco~A. Servedio, and Andrew Wan, \emph{Testing for concise
  representations}, 48th Annual {IEEE} Symposium on Foundations of Computer
  Science {(FOCS} 2007), October 20-23, 2007, Providence, RI, USA, Proceedings,
  2007, pp.~549--558.

\bibitem[GR00]{GoldreichR00}
Oded Goldreich and Dana Ron, \emph{On testing expansion in bounded-degree
  graphs}, Electronic Colloquium on Computational Complexity {(ECCC)}
  \textbf{7} (2000), no.~20.

\bibitem[Hoe]{Hoeffding63}
Wassily Hoeffding, \emph{Probability inequalities for sums of bounded random
  variables}, Journal of the American Statistical Association \textbf{58},
  no.~1963, 13--30.

\bibitem[ILR12]{IndykLR12}
Piotr Indyk, Reut Levi, and Ronitt Rubinfeld, \emph{Approximating and testing
  k-histogram distributions in sub-linear time}, Proceedings of the 31st {ACM}
  {SIGMOD-SIGACT-SIGART} Symposium on Principles of Database Systems, {PODS}
  2012, Scottsdale, AZ, USA, May 20-24, 2012, 2012, pp.~15--22.

\bibitem[Pan08]{Paninski08}
Liam Paninski, \emph{A coincidence-based test for uniformity given very
  sparsely sampled discrete data}, IEEE Transactions on Information Theory
  \textbf{54} (2008), no.~10, 4750--4755.

\bibitem[Ser10]{Servedio10}
Rocco~A. Servedio, \emph{Testing by implicit learning: {A} brief survey},
  Property Testing - Current Research and Surveys [outgrow of a workshop at the
  Institute for Computer Science {(ITCS)} at Tsinghua University, January
  2010], 2010, pp.~197--210.

\bibitem[VV10]{ValiantV10}
Gregory Valiant and Paul Valiant, \emph{A {CLT} and tight lower bounds for
  estimating entropy}, Electronic Colloquium on Computational Complexity
  {(ECCC)} \textbf{17} (2010), 179.

\bibitem[VV14]{ValiantV14}
\bysame, \emph{An automatic inequality prover and instance optimal identity
  testing}, 55th {IEEE} Annual Symposium on Foundations of Computer Science,
  {FOCS} 2014, Philadelphia, PA, USA, October 18-21, 2014, 2014, pp.~51--60.

\end{thebibliography}
\newcommand{\etalchar}[1]{$^{#1}$}
\providecommand{\bysame}{\leavevmode\hbox to3em{\hrulefill}\thinspace}
\providecommand{\MR}{\relax\ifhmode\unskip\space\fi MR }
\providecommand{\MRhref}[2]{%
  \href{http://www.ams.org/mathscinet-getitem?mr=#1}{#2}
}
\providecommand{\href}[2]{#2}

\end{document}